\documentclass[a4paper,USenglish,cleveref,autoref,thm-restate,numberwithinsect,nameinlink]{lipics-v2021}
\hideLIPIcs
\usepackage{graphicx} 
\usepackage{tikz}
\usetikzlibrary{arrows}
\usetikzlibrary{decorations.pathmorphing}
\usetikzlibrary{decorations.markings}

\usepackage[algo2e,vlined,linesnumbered,ruled]{algorithm2e}

\newcommand{\cl}{\ensuremath{\mathcal{L}}}
\newcommand{\R}{\ensuremath{\mathcal{R}}}
\newcommand{\I}{\ensuremath{\mathcal{I}}}
\newcommand{\A}{\ensuremath{\mathcal{A}}}
\renewcommand{\P}{\ensuremath{\mathsf{P}}}
\newcommand{\NP}{\ensuremath{\mathsf{NP}}}
\newcommand{\XP}{\ensuremath{\mathsf{XP}}}
\newcommand{\FPT}{\ensuremath{\mathsf{FPT}}}
\newcommand{\W}{\ensuremath{\mathsf{W[1]}}}
\renewcommand{\O}{\ensuremath{\mathcal{O}}}

\newtheorem{property}{Property}

\crefname{property}{Property}{Properties}
\Crefname{property}{Property}{Properties}
\newtheorem{problem}{Problem}

\theoremstyle{claimstyle}
\newtheorem{claim1}{Claim}
\crefname{claim1}{Claim}{Claims}
\Crefname{claim1}{Claim}{Claims}

\makeatletter
\newcommand{\fullversion}[1]{%
  \@bsphack
  \@esphack
}
\makeatother
\renewcommand{\fullversion}[1]{#1}

\title{Graph Search Trees and the Intermezzo Problem}

\author{Jesse Beisegel}{Institute of Mathematics, Brandenburg University of Technology, Cottbus, Germany}{jesse.beisegel@b-tu.de}{https://orcid.org/0000-0002-8760-0169}{}
\author{Ekkehard Köhler}{Institute of Mathematics, Brandenburg University of Technology, Cottbus, Germany}{ekkehard.koehler@b-tu.de}{}{}
\author{Fabienne Ratajczak}{Institute of Mathematics, Brandenburg University of Technology, Cottbus, Germany}{fabienne.ratajczak@b-tu.de}{https://orcid.org/0000-0002-5823-1771}{}
\author{Robert Scheffler}{Institute of Mathematics, Brandenburg University of Technology, Cottbus, Germany}{robert.scheffler@b-tu.de}{https://orcid.org/0000-0001-6007-4202}{}
\author{Martin Strehler}{Department of Mathematics, Westsächsische Hochschule Zwickau, Zwickau, Germany}{martin.strehler@fh-zwickau.de}{https://orcid.org/0000-0003-4241-6584}{}

\authorrunning{J. Beisegel, E. Köhler, F. Ratajczak, R. Scheffler, and M. Strehler} 

\Copyright{Jesse Beisegel, Ekkehard Köhler, Fabienne Ratajczak, Robert Scheffler, and Martin Strehler} 

\acknowledgements{The authors would like to thank Matja\v{z} Krnc, Martin Milani\v{c}, and Nevena Piva\v{c} for fruitful discussions about first-in and last-in trees as well as about Paul and Mary.}



\nolinenumbers

\begin{document}

\maketitle

\begin{abstract}
The last in-tree recognition problem asks whether a given spanning tree can be derived by connecting each vertex with its rightmost left neighbor of some search ordering. In this study, we demonstrate that the last-in-tree recognition problem for Generic Search is \NP-complete. We utilize this finding to strengthen a complexity result from order theory. Given a partial order $\pi$ and a set of triples, the \NP-complete intermezzo problem asks for a linear extension of $\pi$ where each first element of a triple is not between the other two. We show that this problem remains \NP-complete even when the Hasse diagram of the partial order forms a tree of bounded height. In contrast, we give an \XP-algorithm for the problem when parameterized by the width of the partial order. Furthermore, we show that -- under the assumption of the Exponential Time Hypothesis -- the running time of this algorithm is asymptotically optimal.
\end{abstract}

\section{Introduction}

In the realm of computational combinatorics, one of the primary challenges is to determine a feasible configuration based on incomplete information. This paper aims to elucidate the 
relationships between two notable instances of this problem category: recognition of search trees of graph searches and total ordering with constraints. Specifically, our focus will be on exploring the last-in-tree recognition in the context of generic search and the intermezzo problem.
\subparagraph{Graph Searches.} Graph searches like \emph{Breadth First Search} (BFS) or \emph{Depth First Search} (DFS) are among the most basic algorithms in computer science. Their simplicity belies their significance as they form the backbone of more complex algorithms used to compute key properties of graphs. For instance, DFS can be employed to test for planarity as demonstrated by Hopcraft and Tarjan \cite{hopcrafttarjan74} and 
\emph{Lexicographic Breadth First Search} (LBFS) aids in the recognition and minimum coloring of chordal graphs through a perfect elimination ordering~\cite{rosetarjanlueker76}. Notably, all the above mentioned algorithms operate in linear time, underscoring their efficiency.

In this context, \emph{Generic Search} (GS) represents the most general form of a graph search\fullversion{ (see Figure~\ref{fig:gs})}, with connectivity being its sole constraint: To elaborate, starting from a root vertex $r$, every subsequently visited vertex merely needs to be adjacent to a previously visited vertex. Consequently, GS can yield any total order of the vertices, provided each prefix is connected. A search methodology that bears a close resemblance to GS is the \emph{Maximum Neighborhood Search} (MNS) \cite{corneil2008unified}, which can be perceived as a lexicographic variant of GS. Similarly, BFS and DFS can be implemented by using a queue and a stack, respectively, to store vertices that have not yet been visited.

\fullversion{
\begin{figure}[ht!]
	\centering
	\begin{tikzpicture}[scale=0.5]
	\node (gs) at (3,5.5) {GS};%
	\node (bfs) at (0,4) {BFS};%
	\node (dfs) at (6,4) {DFS};%
	\node (mns) at (3,4) {MNS};%
	\node (mcs) at (3,2) {MCS};%
	\node (lbfs) at (0,2) {LBFS};%
	\node (ldfs) at (6,2) {LDFS};%
	\draw[-stealth'] (gs)--(bfs);%
	\draw[-stealth'] (gs)--(mns);%
	\draw[-stealth'] (gs)--(dfs);%
	\draw[-stealth'] (bfs)--(lbfs);%
	\draw[-stealth'] (mns)--(mcs);%
	\draw[-stealth'] (dfs)--(ldfs);%
	\draw[-stealth'] (mns)--(lbfs);%
	\draw[-stealth'] (mns)--(ldfs);%
	\end{tikzpicture}\caption{Relationships between graph searches. The arrows represent proper inclusions.}\label{fig:gs}
\end{figure}
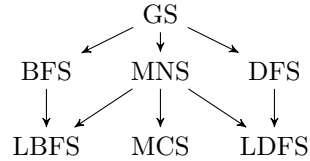
}
\fullversion{\noindent}Recognizing the significance of basic graph search algorithms such as BFS or DFS, recent efforts have been directed towards a deeper understanding of these algorithms. The primary focus of these studies revolves around two structures: end vertices and search trees (for a summary of known results see \cite[Tables~1 and~2]{scheffler2023ready}).
Given a graph $G$ and a specific search rule (e.g., BFS or DFS), the \emph{End Vertex Problem} aims to identify potential final vertices of the search. For GS, solving the end vertex problem is relatively straightforward. As long as a vertex $v$ is not an articulation point, i.e., $G - v$ remains connected, $v$ can serve as an end vertex of GS~\cite{charbithabibmamcarz14}. However, the end vertex problem is \NP-complete for all other common search rules on general graphs \cite{beisegel2019end,charbithabibmamcarz14,corneil2010,zou2022end}. By restricting to special graph classes, linear-time algorithms have been developed to solve this problem, e.g., for BFS on split graphs \cite{charbithabibmamcarz14}, for DFS on interval graphs \cite{beisegel2019end}, and for MNS on chordal graphs~\cite{beisegel2019end}.

Given a graph $G$ and a spanning tree $T$, the \emph{Tree Recognition Problem} seeks to determine whether $T$ can be derived as a search tree. In essence, it questions the feasibility of reconstructing a linear order of vertices from the tree. This problem is typically studied in two variants: first-in-trees and last-in-trees~\cite{beisegel2021recognition}.
In first-in-trees, each vertex is connected to its neighbor that appears first in the search order. Conversely, in last-in-trees (or \cl-trees), each vertex $v$ is a child of its neighbor that appears last before $v$ in the search order. Normally, first-in-trees are used for BFS and last-in-trees for DFS, with existing linear-time algorithms capable of recognizing the corresponding trees in both cases~\cite{hagerup1985biconnected,hagerup1985recognition,korach1989dfs,manber1990recognizing}.
Interestingly, the problem becomes \NP-complete when the search-tree paradigms are swapped between these searches, i.e., using last-in-trees for BFS and first-in-trees for DFS~\cite{scheffler2022recognition}.  Furthermore, Scheffler~\cite{scheffler2023ready} shows that the first-in-tree recognition problem of GS can be solved in linear~time.
\subparagraph{Total Ordering.} A well-known theorem in order theory states that any partial order can be extended to a linear order. This holds true even for infinite sets, as demonstrated by Szpilrajn (Marczewski) through the use of the axiom of choice~\cite{marczewski}.\footnote{He also references unpublished proofs by Banach, Kuratowski, and Tarski.}  The process simplifies considerably for finite sets, where topological sorting algorithms can determine such an extension in linear time~\cite{topologicalsort}.

While partial orders are typically defined by a binary relation, total order problems offer a more general perspective. Here, one is given a set $A$, a family $\cal B$ of subsets $A_i\subseteq A$, and for each $A_i\in\cal B$ one or more valid orderings of the elements within $A_i$. The objective is to ascertain a total order of the elements in $A$ that adheres to all these constraints.

Among the problems, the \emph{Betweenness Problem} and the \emph{Cyclic Ordering Problem} are particularly noteworthy. These two problems have already been discussed in the seminal textbook by Garey and Johnson~\cite{gareyjohnson}. 
In the betweenness problem, we are presented with triples $(a,b,c)$, and the only valid configurations are $a<b<c$ or $c<b<a$. In simpler terms, $b$ must be positioned between $a$ and $c$.
The cyclic ordering problem involves given triples $(a,b,c)$ for which there are three feasible orderings: $a<b<c$, $b<c<a$, or $c<a<b$.
As the appearance in Garey and Johnson’s book already suggests, both of these problems are indeed \NP-complete.    

In~\cite{guttmann2006variations}, Guttmann and Maucher systematically categorized total ordering problems based on pairs and triples. They also introduced the term \emph{Intermezzo} to describe a specific variant: given pairs $(b,c)$ where $b<c$, and triples $(a,b,c)$ where either $a<b<c$ or $b<c<a$, implying that $a$ is not placed between $b$ and $c$. Note that a partial order is defined by both pairs and triples through the relation $b < c$. This problem has been proven to be \NP-complete. 
\subparagraph{Interconnections.} In this context, the problems of identifying end vertices and search trees are interconnected with the total ordering problem, given the underlying vertex order. The end vertex problem asks if a vertex can be the maximal element within this order. On the other hand, the correct search order offers a certificate for the search tree problem that can be checked in linear time. However, the constraints, which include all valid search orders and could potentially be exponential in number, are not explicitly given. Instead, they are implicitly defined by the underlying search paradigm.

Recently, Scheffler~\cite{scheffler2022linearizing} introduced the more general problem of linearizing partial orders where the resultant total order must serve as a search order of a specified graph $G$. He presents polynomial-time algorithms for this problem for several searches and graph classes. In particular, he shows that the problem can be solved for GS on general graphs using a simple greedy algorithm. These results generalize the polynomial-time algorithms for the end vertex problem, given that the partial order can be selected to determine the end vertex.
\subparagraph*{Our Contribution.} After providing the necessary notation, we prove \NP-completeness of the \cl-tree problem for Generic Search (GS) in Section~\ref{sec:ltree}. It is worth noting that two aspects of this result may appear surprising: Firstly, for GS all other problems considered so far can be solved in polynomial time with straightforward methods. Secondly, until now, for any given combination of a search rule (such as BFS, DFS, etc.) and a graph class (like chordal, interval, split, etc.), both tree recognition problems have not been harder than the end vertex problem. Thus, GS on general graphs  represents the first known instance where the end vertex problem is simpler than a tree-recognition problem. 
We use the \NP-completeness of the \cl-tree problem of GS in Section~\ref{sec:intermezzo} to show that the Intermezzo Problem is also \NP-complete even if the partial order $\pi$ is a cs-tree or the height of $\pi$ is bounded. In contrast, we give an \XP-algorithm for the problem when parameterized by the width of $\pi$. Under the assumption of the Exponential Time Hypothesis, we show that the running time of this algorithm is asymptotically optimal.

\section{Preliminaries}\label{sec:prelim}

All the graphs that we consider are simple, finite, non-empty and undirected. Given a graph~$G$, we denote by $V(G)$ the set of \emph{vertices} and by $E(G)$ the set of \emph{edges}.
 
A path $P $ of $G$ is a non-empty subgraph of $G$ with $V(P) = \{v_1, \ldots , v_k\}$ and $E(P) = \{v_1v_2 , \ldots , v_{k-1}v_k\}$, where $v_1, \ldots, v_k$ are all distinct. We will sometimes denote such a path by $v_1 - v_2 - \ldots - v_{k-1} - v_k$.
A graph $G$ is called a \emph{tree} if it is connected and does not contain a cycle. A \emph{spanning tree} $T$ is a subgraph of a graph $G$ which is a tree with $V(T)=V(G)$. A tree together with a distinguished \emph{root vertex} $r$ is said to be \emph{rooted}. In such a rooted tree a vertex $v$ is an \emph{ancestor} of vertex $w$ if $v$ is an element of the unique path from $w$ to the root~$r$. In particular, if $ v $ is adjacent to $ w $, it is called the \emph{parent} of $ w $. Furthermore, a vertex $ w $ is called the \emph{descendant (child)} of $ v $ if $ v $ is the ancestor (parent) of $ w $. We define the \emph{height of a rooted tree} as the maximum number of edges of a path from the root $r$ to any other vertex.
A graph is a \emph{split graph} if its vertex set can be partitioned into a clique and an independent set.

Given a set $X$, a \emph{(binary) relation} $\R$ on $X$ is a subset of the set $X^2 = \{(x,y)~|~x,y \in X\}$. The set $X$ is called the \emph{ground set of \R}. The \emph{reflexive and transitive closure} of a relation $\R$ is the smallest relation $\R'$ such that $\R \subseteq \R'$ and $\R'$ is reflexive and transitive. A \emph{partial order} $ \pi $ on a set $X$ is a reflexive, antisymmetric and transitive relation on $X$. The tuple $(X, \pi)$ is then called a \emph{partially ordered set}. We also denote $(x,y) \in \pi$ by $x \prec_\pi y$ if $x \neq y$.
A \emph{minimal element} of a partial order $\pi$ on $X$ is an element $x \in X$ for which there is no element $y \in X$ with $y \prec_\pi x$.
A \emph{chain} of a partial order $\pi$ on a set $X$ is a set of elements $\{x_1,\ldots,x_k\} \subseteq X$ such that $x_1 \prec_\pi x_2 \prec_\pi \ldots \prec_\pi x_k$. The \emph{height} of $\pi$ is the number of elements of the largest chain of $\pi$. An \emph{antichain} of $\pi$ is a set of elements $\{x_1,\ldots,x_k\} \subseteq X$ such that $x_i \not\prec_\pi x_j$ for any $i, j \in \{1, \ldots, k\}$ 
. The \emph{width} of $\pi$ is the number of elements of the largest antichain of $\pi$.

A \emph{linear ordering} of a finite set $X$ is a bijection $\sigma: X \rightarrow \{1,2,\dots,|X|\}$. We will often refer to linear orderings simply as orderings. Furthermore, we will denote an ordering by a tuple $(x_1, \ldots, x_n)$ which means that $\sigma(x_i) = i$. Given two elements $x$ and $y$ in $X$, we say that $x$ is \emph{to the left} (resp. \emph{to the right}) of $y$ if $\sigma(x)<\sigma(y)$ (resp. $\sigma(x)>\sigma(y)$) and we denote this by $x \prec_{\sigma}y$ (resp.  $x \succ_{\sigma}y$).

A \emph{vertex ordering} of a graph $G$ is a linear ordering of the vertex set $ V(G)$. A vertex ordering $\sigma = (v_1, \ldots, v_n)$ is called \emph{connected} if for any $i \in \{1, \ldots, n\}$ the graph $G[v_1, \ldots, v_i]$ is connected. In this paper, a \emph{graph search} is an algorithm that, given a graph $G$ as input, outputs a connected vertex ordering of $G$.
The graph search that is able to compute any such ordering is called \emph{Generic Search (GS)}.

\section[Complexity of the L-tree Recognition Problem]{Complexity of the \cl-tree Recognition Problem}\label{sec:ltree}

The definition of the term \emph{search tree} varies between different paradigms. However, typically, it consists of the vertices of the graph and, given the search ordering $ (v_1, \ldots , v_n) $, for each vertex $ v_i $ exactly one edge to a $ v_j \in N(v_i) $ with $ j < i $. By specifying to which of the previously visited neighbors a new vertex is adjacent in the tree, we can define different types of graph search trees.
For example, in DFS trees a vertex $ v $ is adjacent to the rightmost neighbor to the left of $ v $. This motivates the following definition.

\begin{definition}
    Given a search ordering $ \sigma:=(v_1, \ldots , v_n) $ of a graph search on a connected graph $ G $, we define the \emph{last-in tree} (or \cl-tree) to be the tree consisting of the vertex set $ V(G) $ and an edge from each vertex $ v_i $ to its rightmost neighbor $ v_j $ in $ \sigma $ with $ j < i $.
\end{definition}
As explained above, for a classical DFS the tree $T$ is an \cl-tree with respect to~$\sigma$. Given this definition, we can state the following decision problem.

\begin{problem}[\cl-Tree Recognition Problem of graph search $\A$]~
    \begin{description}
        \item[Instance:] A connected graph $ G $ and a spanning tree $ T $ of $G$.
        \item[Question:] Is there a graph search ordering of $\A$ such that $ T $ is its \cl-tree of $ G $?
    \end{description}
    \label{prob:tree}
\end{problem}

\noindent Note that we have defined the \cl-tree recognition problem without a given start vertex for the search. It is also possible to define this problem with a fixed start vertex and we call this the \emph{rooted \cl-tree recognition problem}. Obviously, a polynomial-time algorithm for the rooted tree recognition problem yields a polynomial-time algorithm for the general problem by simply repeating the procedure for all vertices. The other direction, however, is not necessarily true.

The \cl-tree recognition problem of GS raised in~\cite{scheffler2022recognition} and~\cite{scheffler2024recognition} is an open problem and in the following we will show that it is in fact \NP-complete. This result will also answer another open question, as Scheffler showed in~\cite{scheffler2022recognition} that the \cl-tree recognition problem of BFS for split graphs is at least as hard as that of GS.

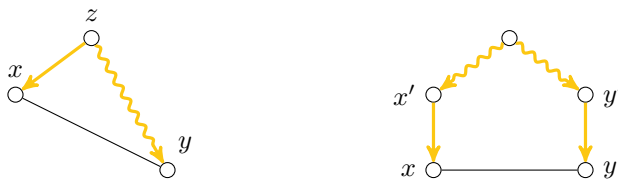
\begin{figure}
\centering
\begin{tikzpicture}[vertex/.style={inner sep=2pt,draw,circle}, yscale=0.5]

\node[vertex, label=$z$] (z) at (1,-0.5) {};
\node[vertex, label= 80:$y$] (y) at (2,-4) {};
\node[vertex, label= $x$] (x) at (0,-2) {};
\draw[lipicsYellow,-stealth',very thick,decorate, decoration={snake,amplitude=.4mm,segment length=2mm,post length=2mm}] (z) --(y); 
\draw[lipicsYellow,-stealth',very thick] (z)--(x);
\draw(x)--(y);

\node[vertex] (r) at (6.5,-0.5) {};
\node[vertex, label=0:$y'$] (y'1) at (7.5,-2) {};
\node[vertex, label= 180:$x$] (x1) at (5.5,-4) {};
\node[vertex, label= 0:$y$] (y1) at (7.5,-4) {};
\node[vertex, label= 180:$x'$] (x'1) at (5.5,-2) {};
\draw[lipicsYellow,-stealth',very thick] (x'1)--(x1);
\draw[lipicsYellow,-stealth',very thick](y'1)--(y1);
\draw(y1)--(x1);

\draw[lipicsYellow,-stealth',very thick,decorate, decoration={snake,amplitude=.4mm,segment length=2mm,post length=2mm}] (r) --(x'1);
\draw[lipicsYellow,-stealth',very thick,decorate, decoration={snake,amplitude=.4mm,segment length=2mm,post length=2mm}] (r) --(y'1);
\end{tikzpicture}
    \caption{On the left is an example of a \emph{hook configuration}. On the right is an example of a \emph{U-bend}. The yellow edges symbolize edges of the spanning tree, the black edges are non-tree edges of the graph and the wavy line represents a directed path in the spanning tree.}\label{fig:hook}
\end{figure}
An important property of \cl-trees of GS can be derived from the non-tree edges that connect vertices of different branches of the tree.

\begin{lemma}\label{lemma:edge}
    Let $T$ be a spanning tree of a graph $G$ rooted in $r$. Let $xy$ be an edge in $E(G) \setminus E(T)$ and let $x'$ and $y'$ be the parents of $x$ and $y$ in $T$, respectively. If $T$ is an \cl-tree of a GS ordering $\sigma$ starting with $r$, then it either holds that $x' \prec_\sigma x \prec_\sigma y' \prec_\sigma y$ or $y' \prec_\sigma y \prec_\sigma x' \prec_\sigma x$.
\end{lemma}
\fullversion{\begin{proof}
    Suppose that $x' \prec_\sigma y'$. If $y' \prec_\sigma x \prec_\sigma y$, then $y'$ would not be the parent of $y$ in $T$; a contradiction. Otherwise, if $y' \prec_\sigma y \prec_\sigma x $, then $x'$ would not be the parent of $x$ in $T$; again a contradiction. Therefore, we see with the orders implied by the tree edges that $x' \prec_\sigma x \prec_\sigma y' \prec_\sigma y$. Supposing that $y' \prec_\sigma x'$, we see by symmetry that $y' \prec_\sigma y \prec_\sigma x' \prec_\sigma x$, which proves the lemma.
\end{proof}}
The configuration described in~\Cref{lemma:edge} will be called a \emph{U-bend configuration} (see~\Cref{fig:hook} on the right).

Before we begin with our main results, we should analyze examples of some rooted spanning trees that cannot be \cl-trees of GS. One of the smallest examples can be found to the left in~\Cref{fig:counter-example}. The example on the right is a generalization with arbitrary many branches of the spanning tree. These examples can be easily described using a concept called \emph{hook configuration}. This is a special case of a U-bend where the parent of one vertex is an ancestor of the others.

\begin{definition}
    Let $T$ be a spanning tree of a graph $G$ rooted in $r \in V(G)$. We say that a triple of vertices $x$, $y$, and $z$ forms a \emph{hook configuration} or a \emph{hook}  if $z$ is the parent of $x$ in $T$, $xy \in E(G) \setminus E(T)$ and $y$ is a descendant of $z$ but $y$ is not a descendant of $x$ (see \Cref{fig:hook} on the left). We call $x$ the \emph{point} and $y$ the \emph{eye} of the hook.
\end{definition}
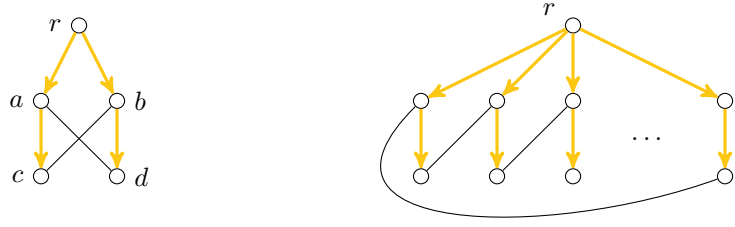
\begin{figure}
    \centering
    \begin{tikzpicture}[vertex/.style={inner sep=2pt,draw,circle}]
        \node[vertex, label=180:$r$] (r) at (1.5,3) {};
        \node[vertex, label= 180:$a$] (a) at (1,2) {};
        \node[vertex, label=0:$b$] (b) at (2,2) {};
        \node[vertex, label=180:$c$] (c) at (1,1) {};
        \node[vertex, label=0:$d$] (d) at (2,1) {};
        \draw[lipicsYellow,-stealth',very thick] (r) -- (a);
        \draw[lipicsYellow,-stealth',very thick] (a) -- (c);
        \draw[lipicsYellow,-stealth',very thick] (r) -- (b);
        \draw[lipicsYellow,-stealth',very thick] (b) -- (d);
        \draw (a) -- (d);
        \draw (b) -- (c);

        \node[vertex,label=170:$r$] (r1) at (8,3) {};
        \node[vertex] (a1) at (6,2) {};
        \node[vertex] (b1) at (7,2) {};
        \node[vertex] (c1) at (6,1) {};
        \node[vertex] (d1) at (7,1) {};
        \node[vertex] (e1) at (8,2) {};
        \node[vertex] (f1) at (8,1) {};
        \node[vertex] (g1) at (10,2) {};
        \node[vertex] (h1) at (10,1) {};
        \draw[lipicsYellow,-stealth',very thick] (r1) -- (a1);
        \draw[lipicsYellow,-stealth',very thick] (a1) -- (c1);
        \draw[lipicsYellow,-stealth',very thick] (r1) -- (b1);
        \draw[lipicsYellow,-stealth',very thick] (b1) -- (d1);
        \draw[lipicsYellow,-stealth',very thick] (r1) -- (e1);
        \draw[lipicsYellow,-stealth',very thick] (e1) -- (f1);
        \draw[lipicsYellow,-stealth',very thick] (r1) -- (g1);
        \draw[lipicsYellow,-stealth',very thick] (g1) -- (h1);
       \path (8,1.5) -- (10,1.5) node [midway,sloped] {$\dots$};
        \draw (b1) -- (c1);
        \draw (e1) -- (d1);
        \draw[looseness=1.5,out=225,in=200] (a1) to (h1);
    \end{tikzpicture}
    \caption{Family of graphs where the rooted spanning trees (yellow edges) are not \cl-trees of GS.}
    \label{fig:counter-example}
\end{figure}
These hook configurations have a strong a priori effect on the sequence of any search ordering corresponding to that tree.

\begin{lemma}\label{lemma:hook}
    Let $x$ and $y$ be part of a hook configuration of $T$ rooted in $r \in V(G)$ with point $x$ and eye $y$. Then for any GS ordering $\sigma$ starting in $r$ with \cl-tree $T$ it holds that $x \prec_\sigma y$.
\end{lemma}
\fullversion{\begin{proof}
   Let $z$ be the parent of $x$. It holds that $z \prec_\sigma y$ since $y$ is a descendant of $z$ in $T$. If $y \prec_\sigma x$, then $xz$ would not have been chosen as a tree edge, since $y$ was chosen after $z$; this is a contradiction to the assumption that $z$ is the parent of $x$.
\end{proof}}
For the examples shown in~\Cref{fig:counter-example}, it is possible to see that the hook configurations create something like a cycle in the ordering using~\Cref{lemma:hook}: We see that $a \prec d$ and $b \prec c$ and these contradict each other because of the tree edges.

In the special case that the graph together with its spanning tree does not contain any hooks, it is trivial to decide the \cl-Tree Recognition Problem.

\begin{theorem}\label{theorem:no-hook}
    Let $T$ be a spanning tree of a graph $G$ rooted in $r \in V(G)$. If there is no hook configuration, then any DFS ordering of $T$ starting in $r$ is a GS ordering of $G$ with $\cl$-tree $T$. Therefore, any such tree together with $G$ is a \textsc{Yes}-instance for the \cl-Tree Recognition Problem of Generic Search.
\end{theorem}
\begin{proof}
    Let $\sigma$ be a DFS ordering of $T$ starting in $r$. This ordering $\sigma$ fulfills the following property also called \emph{four point condition} (see for example~\cite{corneil2008unified}): If $a \prec_\sigma b \prec_\sigma c$ and $ac \in E$ and $ab \notin E$, then there exists a vertex $d$ with $a \prec_\sigma d \prec_\sigma b$ such that $db \in E$.

    Suppose that $\sigma$ does not induce the \cl-tree $T$ for $G$. Let $w$ be the leftmost vertex in $\sigma$ such that there exist $u$ and $v$ with $u \prec_\sigma v \prec_\sigma w$ with $uw \in E(T)$ and $vw \in E(G) \setminus E(T)$. If $uv$ is an edge in the tree $T$, then $u,v$ and $w$ form a hook configuration; a contradiction to the assumption. Therefore, we can assume that $uv \notin E(T)$. Now we can apply the four point condition to vertices $u,v$ and $w$ (note that the DFS was executed on $T$). A result of the four point condition, is the fact that there must exist a $u$-$v$-path $u = d_1- \dots - d_k = v$ with $u \prec_\sigma d_2 \prec_\sigma \dots \prec_\sigma d_{k-1}\prec_\sigma v$ (Corollary 2.6 in~\cite{corneil2008unified}). In particular, we see that $v$ is a descendant of $u$. Together with the fact that $uw \in T$ and $vw \in E(G) \setminus E(T)$, we see that $u,v$ and $w$ form a hook configuration; a contradiction to the assumptions of the theorem. This implies that each vertex in the search is connected to its correct parent in $T$, proving that $T$ is an \cl-tree of GS in $G$.
\end{proof}
This theorem could lead to the assumption that deciding whether a given spanning tree is an $\cl$-tree of GS only amounts to an analysis of all the hook configurations. In fact, it is easy to see that we can find all hook configurations in polynomial time.
However, it is not always so simple. In fact, we show in the following that in general the \cl-Tree Recognition Problem of GS is \NP-complete.
We describe a reduction from 3-SAT, i.e., we are given an instance $\mathcal{I}$ of 3-SAT and derive an instance $(G(\mathcal{I}),T(\mathcal{I}))$ for the \cl-Tree Recognition Problem of~GS.

Let $\mathcal{I}$ be an instance of 3-SAT with variable set $\{X_1, \ldots , X_n\}$ and clause set $\{C^1, \ldots , C^m\}$. For ease of notation, we define the positive literal $X_j$ as $X_j(1)$ and the negative literal $\lnot X_j$ as~$X_j(0)$. First we define the spanning tree $T(\mathcal{I})$ and then we add the edges missing to give the full graph $G(\mathcal{I})$. For each variable $X_j$ of $\mathcal{I}$ we add two vertices $x_j(0)$ and $x_j(1)$ (representing the two literals of $X_j$) to $V(G(\mathcal{I}))$. These vertices are all adjacent to a common root vertex $r$. Furthermore, $r$ is adjacent to the clause-hub-vertex~$C$. Now we add a vertex $c^i$ for each clause $C^i$ and connect it to $C$ (see~\cref{fig:variable} for a depiction of this setup.). Furthermore, for each occurrence of a literal associated with vertex $x_j(p)$ in a clause $C^i$ we add a vertex~$x_j^i$. As we may assume that only one of the literals appears in a clause, we do not have to use an index to mark whether the vertex belongs to a negated variable or not. This sums up the basic setup concerning the variables and clauses. However, for technical reasons we need several more vertices (see~\cref{fig:clause}):
\begin{itemize}
    \item For each $c^i$ we add the vertices $a^i_0, a^i_1, a^i_2,b^i_0, b^i_1, b^i_2$ in $T(\mathcal{I})$. These are called the \emph{technical vertices of} $c^i$. 
    \item These vertices form the paths $c^i - a^i_0 - b^i_0$, $c^i - a^i_1 - b^i_1$ and $c^i - a^i_2 - b^i_2$ in $T(\mathcal{I})$.
    \item For each $x^i_j$ we add vertices $d^i_j,e^i_j,f^i_j$. These are called the \emph{technical vertices of} $x^i_j$.
    \item These vertices form the path $x^i_j-d^i_j-e^i_j-f^i_j$ in $T(\mathcal{I})$. 
\end{itemize}

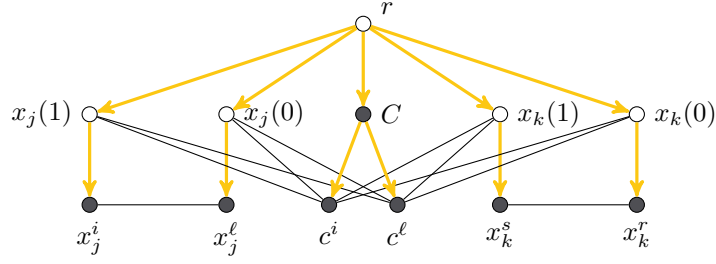
\begin{figure}
    \centering
    \begin{tikzpicture}[vertex/.style={inner sep=2pt,draw,circle},scale =0.6]
        \node[vertex, label=5:$r$] (r) at (0,0) {};
        \node[vertex, label=180:$x_j(1)$] (x1) at (-6, -2) {};
        \node[vertex, label=0:$x_j(0)$] (nx1) at (-3, -2) {};
        \node[vertex, fill=lipicsLineGray, label=0:$C$] (C) at (0, -2) {};
        \node[vertex, fill=lipicsLineGray, label=-90:$x_j^i$] (x1i) at (-6,-4) {};
        \node[vertex, fill=lipicsLineGray, label=-90:$ x_j^\ell$] (nx1i) at (-3,-4) {};
        \node[vertex, fill=lipicsLineGray, label=-90:$c^i$] (ci) at (-0.75,-4) {};

        \node[vertex, label=0:$x_k(1)$] (x2) at (3, -2) {};
        \node[vertex, label=0:$x_k(0)$] (nx2) at (6, -2) {};
        \node[vertex, fill=lipicsLineGray, label=-90:$x_k^s$] (x2i) at (3,-4) {};
        \node[vertex, fill=lipicsLineGray, label=-90:$ x_k^r$] (nx2i) at (6,-4) {};
        \node[vertex, fill=lipicsLineGray, label=-90:$c^\ell$] (cj) at (0.75,-4) {};

        \draw[lipicsYellow, very thick, -stealth'] (r) -- (x1);
        \draw[lipicsYellow, very thick, -stealth'] (r) -- (nx1);
        \draw[lipicsYellow, very thick, -stealth'] (r) -- (x2);
        \draw[lipicsYellow, very thick, -stealth'] (r) -- (nx2);
        \draw[lipicsYellow, very thick, -stealth'] (r) -- (C);
        \draw[lipicsYellow, very thick, -stealth'] (x1) -- (x1i);
        \draw[lipicsYellow, very thick, -stealth'] (nx1) -- (nx1i);
        \draw[lipicsYellow, very thick, -stealth'] (x2) -- (x2i);
        \draw[lipicsYellow, very thick, -stealth'] (nx2) -- (nx2i);

        \draw[black] (x1) -- (ci);
        \draw[black] (nx1) -- (ci);
        \draw[black] (x1) -- (cj);
        \draw[black] (nx1) -- (cj);
        \draw[black] (x2) -- (ci);
        \draw[black] (nx2) -- (ci);
        \draw[black] (x2) -- (cj);
        \draw[black] (nx2) -- (cj);
        \draw[black] (x1i) -- (nx1i);
        \draw[black] (x2i) -- (nx2i);

        \draw[lipicsYellow, very thick, -stealth'] (C) -- (ci);
        \draw[lipicsYellow, very thick, -stealth'] (C) -- (cj);
    \end{tikzpicture}
    \caption{This figure illustrates the variable gadget. The yellow edges are the tree edges, and the vertices marked in gray appear again in the clause gadget. The literal $X_j$ appears in clause $C^k$, the literal $\lnot X_j$ appears in clause $C^\ell$. Note that for the $x_j^i$ vertices we do not need to denote whether they belong to $X_j$ or its negation, as each clause only contains either $X_j$ or $\lnot X_j$.}
        \label{fig:variable}
        
\end{figure}

\noindent This concludes the definition of the tree $T(\mathcal{I})$, to which we will now add the remaining edges for $G(\mathcal{I})$ (see~\cref{fig:variable} and~\cref{fig:clause}):
\begin{itemize}
    \item For each $X_j$ we add edges $x_j(0)c^i$ and $x_j(1)c^i$ for $i \in \{1, \ldots, m\}$. 
    \item For each $x_j^i$ adjacent to $x_j$ we add an edge $x_j^i x_j^k$ if both $X_j$ occurs in clause $C^i$ and $\lnot X_j$ occurs in literal $C^k$.
    \item For any clause $C^i =\{X_{j_0}(p_0), X_{j_1}(p_1), X_{j_2}(p_2)\}$ we add the edges:
    \begin{itemize}
        \item $Ce^i_{j_p}$ for $p \in \{0,1,2\}$.
        \item $c^i e^i_{j_p}$ for $p \in \{0,1,2\}$.
        \item $a^i_p f^i_{j_p}$ for $p \in \{0,1,2\}$.
        \item $b^i_p e^i_{j_p}$ for $p \in \{0,1,2\}$.
        \item $b^i_p d^i_{j_{(p+1) \bmod 3}}$ for $p \in \{0,1,2\}$.
    \end{itemize} 
\end{itemize}

\noindent The modulo operation applied to the indices to define the edges in the clause gadget illustrates the circularity inherent in that gadget. Visiting one of the branches below $c^i$ effectively unlocks one of the branches below an $x^i_j$. Conversely, visiting an $x^i_j$ before $c^i$ blocks a corresponding branch of $c^i$. This effect is what leads to the property that we will show in~\cref{lem:clause}.

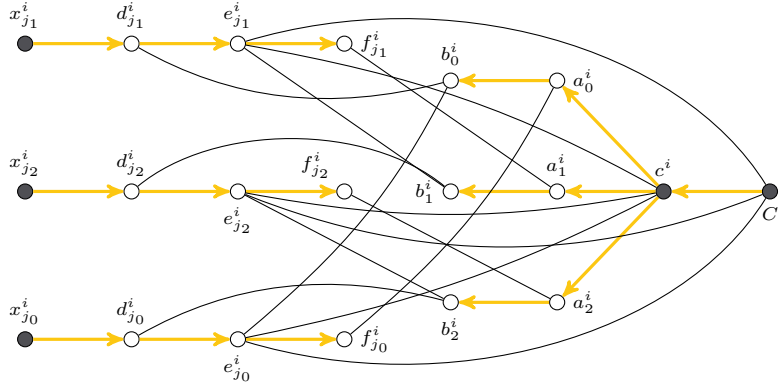
\begin{figure}
    \centering
    \begin{tikzpicture}[vertex/.style={inner sep=2pt,draw,circle},yscale = 0.7,scale =0.7,rotate=90]
    \scriptsize
        \node[vertex, fill=lipicsLineGray, label=90:$x^i_{j_0}$] (x1) at (0,0) {};
        \node[vertex, fill=lipicsLineGray, label=90:$x^i_{j_2}$] (y1) at (4,0) {};
        \node[vertex, fill=lipicsLineGray, label=90:$x^i_{j_1}$] (z1) at (8,0) {};
        \node[vertex, label=90:$d^i_{j_0}$] (x2) at (0,-2) {};
        \node[vertex, label=90:$d^i_{j_2}$] (y2) at (4,-2) {};
        \node[vertex, label=90:$d^i_{j_1}$] (z2) at (8,-2) {};
        \node[vertex, label=270:$e^i_{j_0}$] (x3) at (0,-4) {};
        \node[vertex, label=270:$e^i_{j_2}$] (y3) at (4,-4) {};
        \node[vertex, label=90:$e^i_{j_1}$] (z3) at (8,-4) {};
        \node[vertex, label=0:$f^i_{j_0}$] (x4) at (0,-6) {};
        \node[vertex, label=135:$f^i_{j_2}$] (y4) at (4,-6) {};
        \node[vertex, label=0:$f^i_{j_1}$] (z4) at (8,-6) {};
\begin{scope}[yscale=1]
        \node[vertex, fill=lipicsLineGray, label=-90:$C$] (c) at (4,-14) {};
        \node[vertex, fill=lipicsLineGray, label=90:$c^i$] (ci) at (4,-12) {};
        \node[vertex,label=0:$a^i_2$] (d1) at (1,-10) {};
        \node[vertex,label=90:$a^i_1$] (e1) at (4,-10) {};
        \node[vertex,label=0:$a^i_0$] (f1) at (7,-10) {};
        \node[vertex,label=270:$b^i_2$] (d2) at (1,-8) {};
        \node[vertex,label=180:$b^i_1$] (e2) at (4,-8) {};
        \node[vertex,label=90:$b^i_0$] (f2) at (7,-8) {};
\end{scope}
        \draw[lipicsYellow, -stealth', very thick] (x1) -- (x2);
        \draw[lipicsYellow, -stealth', very thick] (x2) -- (x3);
        \draw[lipicsYellow, -stealth', very thick] (x3) -- (x4);
        \draw[lipicsYellow, -stealth', very thick] (y1) -- (y2);
        \draw[lipicsYellow, -stealth', very thick] (y2) -- (y3);
        \draw[lipicsYellow, -stealth', very thick] (y3) -- (y4);
        \draw[lipicsYellow, -stealth', very thick] (z1) -- (z2);
        \draw[lipicsYellow, -stealth', very thick] (z2) -- (z3);
        \draw[lipicsYellow, -stealth', very thick] (z3) -- (z4);

        \draw[lipicsYellow, -stealth', very thick] (c) -- (ci);
        \draw[lipicsYellow, -stealth', very thick] (ci) -- (d1);
        \draw[lipicsYellow, -stealth', very thick] (ci) -- (e1);
        \draw[lipicsYellow, -stealth', very thick] (ci) -- (f1);
        \draw[lipicsYellow, -stealth', very thick] (d1) -- (d2);
        \draw[lipicsYellow, -stealth', very thick] (e1) -- (e2);
        \draw[lipicsYellow, -stealth', very thick] (f1) -- (f2);

        \draw[black,  bend left] (x2) to (d2);
        \draw[black,  bend left= 45] (y2) to (e2);
        \draw[black,  bend right] (z2) to (f2);
        \draw[black] (d1) to (y4);
        \draw[black] (d2) to (y3);
        \draw[black] (e1) to (z4);
        \draw[black] (e2) to (z3);
        \draw[black,  bend left=10] (f1) to (x4);
        \draw[black, bend left=10] (f2) to (x3);

        \draw[black,  bend angle=45, bend left] (c) to (x3);
        \draw[black,  bend left] (c) to (y3);
        \draw[black,  bend angle=45, bend right] (c) to (z3);
        \draw[black,  bend left=10] (ci) to (x3);
        \draw[black,  bend angle=15, bend left] (ci) to (y3);
        \draw[black,  bend right=12] (ci) to (z3);
    \end{tikzpicture}
    \caption{This figure illustrates the clause gadget. The tree edges are colored yellow and the gray vertices mark the vertices that can be found in the variable gadget. Note that we have drawn this figure horizontally to make the embedding cleaner. The directions of the tree edges denote the direction from the root of the tree.}
    \label{fig:clause}
\end{figure}

The first step in our reduction is to check whether we can use a search ordering that achieves $T(\I)$ to construct an assignment of the 3-SAT instance $\I$. For each variable we will choose the assignment that corresponds to the literal vertex chosen \emph{second}, i.e., if $x_j(1)$ is chosen first we assign $X_j$ the value~0 and if $x_j(0)$ is chosen first we assign $X_j$ the value~1. The following lemma shows that the children of the literal vertex that is chosen first are visited before the clause vertices. 

\begin{lemma}\label{lem:hub}
    If $T(\I)$ is an \cl-tree of some GS ordering $\sigma$ of $G(\I)$ starting in $r$, then it holds for every variable $x_j$ that all the children of $x_j(1)$ or all the children of $x_j(0)$ are to the left of vertex $C$ in $\sigma$. In particular, if $x_j(p) \prec_\sigma x_j(q)$, then the children of $x_j(p)$ are all to the left of $x_j(q)$.
\end{lemma}
\begin{proof}
    The sets of vertices $\{x_j(1),c^i,r\}$ and $\{x_j(0),c^i,r\}$ form hook configurations with eye $c^i$ and point $x_j(1)$ or $x_j(0)$, respectively. By \cref{lemma:hook}, it holds that $x_j(1) \prec_\sigma c^i$ and $x_j(0) \prec_\sigma c^i$. Furthermore, with \cref{lemma:edge} we see that $x_j(1) \prec_\sigma C$ and $x_j(0) \prec_\sigma C$, as for example $r,x_j(1),c^i$ and $C$ form a U-bend. W.l.o.g., we may assume that $x_j(1) \prec_\sigma x_j(0)$. Let $u$ be an arbitrary child of $x_j(1)$ and let $v$ be an arbitrary child of $x_j(0)$. By construction of $G(\I)$ and $T(\I)$, we know that $uv \in E(G(\I)) \setminus E(T(\I))$. Then, due to \cref{lemma:edge} and our assumption, we see that $u \prec_\sigma x_j(0) \prec_\sigma c_i$. Therefore, the children of $x_j(1)$ are to the left of $x_j(0)$ and thus also to the left of the clause hub vertex $C$, proving the statement of the lemma.
\end{proof}
The next lemma motivates how the choice of the variable assignment is used to check whether a given clause is fulfilled.

\begin{lemma}\label{lem:clause}
    If $T(\I)$ is an \cl-tree of $G(\I)$ for the search ordering $\sigma$, then for each clause $C^i =\{X_{j_0}(p_0), X_{j_1}(p_1), X_{j_2}(p_2)\}$ it holds that $c^i \prec_\sigma x_{j_q}^i$ for some $q \in \{0,1,2\}$.
\end{lemma}
\begin{proof}
    Assume for contradiction that $c^i$ is to the right of $x_{j_q}^i$ for all $q \in \{0,1,2\}$. Using the hook and U-bend rules from~\cref{lemma:hook,lemma:edge}, we can show that $T(\I)$ is not an \cl-tree for $\sigma$. For all $q \in \{0,1,2\}$, the set $\{C, e^i_{j_q}, r\}$ forms a hook with point $C$ and eye $e^i_{j_q}$ and, thus, $C \prec_\sigma e^i_{j_q}$. Now the U-bend induced by the edge $e^i_{j_q}c^i$ implies that $c^i \prec_\sigma e^i_{j_q}$. For any 
    $q \in \{0,1,2\}$, the vertex $d^i_{j_q}$ is adjacent to some vertex $b^i_{q'}$ and the corresponding edge induces a U-bend. Since all three vertices $x_{j_q}^i$ are to the left of all three vertices $a^i_{q'}$, these U-bends imply that at least one vertex $d^i_{j_q}$ is to the left of all $a^i_{q'}$. Fix that vertex $d^i_{j_q}$. Now the edge $b^i_qe^i_{j_q}$ implies that $e^i_{j_q} \prec_\sigma a^i_q$. Summarizing, $c^i \prec_\sigma e^i_{j_q} \prec_\sigma a^i_q$. However, this contradicts \cref{lemma:edge} as the edge $f^i_{j_q}a^i_q$ induces a U-bend where both parents ($c^i$ and $e^i_{j_q}$) are to the left of both children ($a^i_q$ and $f^i_{j_q}$). This concludes the proof.    
\end{proof}
This lemma shows that for each clause there is one literal for which the corresponding child belonging to that clause has to be chosen after the clause vertex. This will be the literal that satisfies the clause in a fulfilling assignment. On the other hand, if there is a literal whose child is chosen after the clause vertex, then the corresponding clause gadget can be correctly traversed.
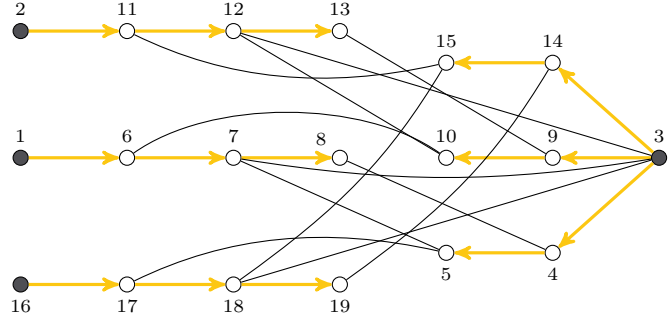
\begin{figure}
    \centering
    \begin{tikzpicture}[vertex/.style={inner sep=2pt,draw,circle},yscale = 0.6,scale =0.7,rotate=90]
    \scriptsize
        \node[vertex, fill=lipicsLineGray, label=270:$1$] (x1) at (0,0) {};
        \node[vertex, fill=lipicsLineGray, label=90:$16$] (y1) at (4,0) {};
        \node[vertex, fill=lipicsLineGray, label=90:$2$] (z1) at (8,0) {};
        \node[vertex, label=270:$11$] (x2) at (0,-2) {};
        \node[vertex, label=90:$17$] (y2) at (4,-2) {};
        \node[vertex, label=90:$6$] (z2) at (8,-2) {};
        \node[vertex, label=270:$12$] (x3) at (0,-4) {};
        \node[vertex, label=90:$18$] (y3) at (4,-4) {};
        \node[vertex, label=90:$7$] (z3) at (8,-4) {};
        \node[vertex, label=270:$13$] (x4) at (0,-6) {};
        \node[vertex, label=90:$19$] (y4) at (4,-6) {};
        \node[vertex, label=90:$8$] (z4) at (8,-6) {};
\begin{scope}[yscale=1]
        \node[vertex, fill=lipicsLineGray, label=90:$3$] (ci) at (4,-12) {};
        \node[vertex,label=270:$14$] (d1) at (1,-10) {};
        \node[vertex,label=90:$4$] (e1) at (4,-10) {};
        \node[vertex,label=90:$9$] (f1) at (7,-10) {};
        \node[vertex,label=270:$15$] (d2) at (1,-8) {};
        \node[vertex,label=90:$5$] (e2) at (4,-8) {};
        \node[vertex,label=90:$10$] (f2) at (7,-8) {};
\end{scope}
        \draw[lipicsYellow, -stealth', very thick] (x1) -- (x2);
        \draw[lipicsYellow, -stealth', very thick] (x2) -- (x3);
        \draw[lipicsYellow, -stealth', very thick] (x3) -- (x4);
        \draw[lipicsYellow, -stealth', very thick] (y1) -- (y2);
        \draw[lipicsYellow, -stealth', very thick] (y2) -- (y3);
        \draw[lipicsYellow, -stealth', very thick] (y3) -- (y4);
        \draw[lipicsYellow, -stealth', very thick] (z1) -- (z2);
        \draw[lipicsYellow, -stealth', very thick] (z2) -- (z3);
        \draw[lipicsYellow, -stealth', very thick] (z3) -- (z4);

        \draw[lipicsYellow, -stealth', very thick] (ci) -- (d1);
        \draw[lipicsYellow, -stealth', very thick] (ci) -- (e1);
        \draw[lipicsYellow, -stealth', very thick] (ci) -- (f1);
        \draw[lipicsYellow, -stealth', very thick] (d1) -- (d2);
        \draw[lipicsYellow, -stealth', very thick] (e1) -- (e2);
        \draw[lipicsYellow, -stealth', very thick] (f1) -- (f2);

        \draw[black,  bend left] (x2) to (d2);
        \draw[black,  bend left= 45] (y2) to (e2);
        \draw[black,  bend right] (z2) to (f2);
        \draw[black] (d1) to (y4);
        \draw[black] (d2) to (y3);
        \draw[black] (e1) to (z4);
        \draw[black] (e2) to (z3);
        \draw[black,  bend left=10] (f1) to (x4);
        \draw[black,  bend left=10] (f2) to (x3);

        \draw[black,  bend left=0] (ci) to (x3);
        \draw[black,  bend angle=15, bend left] (ci) to (y3);
        \draw[black,  bend right=0] (ci) to (z3);

    \end{tikzpicture}
    \caption{Example of an ordering of the technical vertices for a clause where two literals (1 and~2) were chosen to be false. Compare with~\cref{fig:clause}.}
    \label{fig:technical}
\end{figure}
Combining these results proves the main result of this section.
\begin{theorem}\label{thm:main}
    The rooted \cl-tree recognition problem of Generic Search is \NP-complete for rooted spanning trees of height $5$.
\end{theorem}

\begin{proof}
    Let $\sigma$ be a GS ordering with \cl-tree $T(\I)$. We define an assignment $\A$ by setting any variable to false if and only if the positive literal appears before its negative literal in $\sigma$. We claim that this is a fulfilling assignment. Clearly we only need to show that for each clause at least one literal was chosen to be true. In particular, by~\cref{lem:clause} we see that for each $C^i =\{X_{j_0}(p_0), X_{j_1}(p_1), X_{j_2}(p_2)\}$ at least one vertex $x^i_{j_q}$ is to the right of $c^i$ and, thus, to the right of the vertex $C$. Due to \cref{lem:hub}, the parent of $x^i_{j_q}$ is to the right of the vertex of variable. This implies that at least one literal contained in $C^i$ is fulfilled in $\A$.

    Let $\cal{A}$ be a fulfilling assignment of $\I$. We now show that in this case we can construct a GS ordering of $G(\I)$ which has the \cl-tree $T(\I)$.
    The broad idea is to choose the literals that are false for $\A$ first followed by their children. Then we choose the literals that are true, followed by the clause hub vertex $C$ and then the clause vertices. Finally, we need to visit the technical vertices in the correct order, followed by the descendants of the true literals. In the following, we define several suborders that need to be combined into the final linear ordering $\sigma$.

    As explained above, we begin the search ordering $\sigma$ by visiting the root $r$  and then all literals that are set to false by $\A$ in arbitrary order. Next we visit all children of these vertices. In the next phase, we visit all remaining literal vertices in an arbitrary order. At this point, we can visit the clause hub vertex $C$ (see~\cref{lem:hub}) followed by the clause vertices in arbitrary order. Let $C^i =\{X_{j_0}(p_0), X_{j_1}(p_1), X_{j_2}(p_2)\}$ be some clause. If all literals of $C^i$ were chosen to be true (i.e., all the vertices $x^i_{j_0}$, $x^i_{j_1}$, and $x^i_{j_2}$ are to the right of $c^i$), then we can visit the technical vertices of $c^i$ in the order $a^i_0, a^i_1, a^i_2,b^i_0, b^i_1, b^i_2$ (or any other order that conforms with GS) followed by the literal vertices and their technical vertices following the order that is implied by the tree edges. If exactly one of the literals, say w.l.o.g. $X_{j_0}(p_0)$, is chosen to be false by $\A$, then we use the order $x^i_{j_0},c^i, a^i_0, b^i_0, d^i_{j_0}, e^i_{j_0},f^i_{j_0}$. Then we visit the remaining technical vertices of $c^i$ followed by $x^i_{j_1}, x^i_{j_2}$ and their technical vertices.
    
    If exactly two of the literals, say w.l.o.g. $X_{j_0}(p_0)$ and $X_{j_1}(p_1)$, are chosen to be false by $\A$, then we use the order $x^i_{j_0},x^i_{j_1},c^i, a^i_1, b^i_1, d^i_1, e^i_1,f^i_1, a^i_0,b^i_0,d^i_0, e^i_0,f^i_0$. Then we visit the remaining technical vertices of $c^i$ followed by $x^i_{j_2}$ and its technical vertices (see \cref{fig:technical} for an illustration).

    Because $\A$  is a fulfilling assignment, we know that each clause has at most two literals that are set to false. Therefore, we can combine all of these orderings to a comprehensive GS ordering $\sigma$ and confirm that $T(\I)$ is in fact an \cl-tree of $\sigma$. 
\end{proof}
Using~{\cite[Theorems 21 and 23]{scheffler2022recognition}},~we can strengthen \cref{thm:main} to the case of split graphs and give a similar result for Breadth First Search. 
\begin{corollary}
    The rooted \cl-tree recognition problems of Generic Search and Breadth First Search are \NP-complete even if the input is restricted to split graphs and to rooted spanning trees of height $12$.
\end{corollary}
Note that these results are all for the \emph{rooted} \cl-tree recognition problem. Using a small gadget, we can also extend the hardness of the GS \cl-tree recognition to the unrooted problem.

\begin{corollary}
    The \cl-tree recognition problem of Generic Search is \NP-complete.
\end{corollary}
\begin{proof}
    Let $G$ be some graph and $T$ be a spanning tree of $G$ with root $r$. We add three vertices $a,b,c$ to $G$ in the following way to form a new graph $G'$: Let $V(G')= V(G) \cup \{a,b,c\}$. Furthermore, $E(G') = E(G)\cup \{ab,ac,bc,ar\}$. Finally, we define a spanning tree $T'$ of $G'$ with $V(T') = V(G')$ and $E(T') = E(T) \cup \{ar,ab,ac\}$.

    Due to the conflicting hooks among $a,b,c$, either $b$ or $c$ must be visited before $a$, if $T'$ is to be an \cl-tree of $G'$. This makes $r$ the de facto root of $T'-\{a,b,c\}$, showing that the rooted \cl-tree recognition problem for $G$ and $T$ is equivalent to the unrooted one for $G'$ and~$T'$.
\end{proof}


\section{The Intermezzo Problem}\label{sec:intermezzo}

Given a rooted spanning tree $T$, the basic property that has to be fulfilled by a vertex ordering $\sigma$ for it to have $T$ as an \cl-tree is the following:
If there is a vertex $z$ with parent $y$ and $z$ has a non-tree edge to vertex $x$, then $x$ is not allowed to be between $y$ and $z$ in the vertex ordering. These constraints are similar to those used in the following problem introduced by Guttmann and Maucher~\cite{guttmann2006variations}.   

\begin{problem}[General Intermezzo Problem]~
    \begin{description}
        \item[Instance:] Finite set $A$, set $C$ of triples of distinct elements of $A$
        \item[Question:] Is there an ordering of $A$ such that for all $(x,y,z) \in C$ it holds that $x \prec_\sigma y \prec_\sigma z$ or $y \prec_\sigma z \prec_\sigma x$?
    \end{description}
    \label{prob:ordering1}
\end{problem}

\noindent We call an ordering that fulfills the constraints of $C$ an \emph{intermezzo ordering}. Note that Guttmann and Maucher do not give a name for the general problem as they introduce it as one case of a large family of constrained ordering problems. We derived the name from the more restricted \emph{Intermezzo problem}. This problem additionally forces the triples in $C$ to be pairwise disjoint.

\begin{problem}[Intermezzo Problem~\cite{guttmann2006variations}]~
    \begin{description}
        \item[Instance:] Finite set $A$, set $B$ of pairs of $A$, set $C$ of pairwise disjoint triples of distinct elements of $A$.
        \item[Question:] Is there an ordering of $A$ such that for all $(x,y) \in B$ it holds that $x \prec_\sigma y$ and for all $(x,y,z) \in C$ it holds that $x \prec_\sigma y \prec_\sigma z$ or $y \prec_\sigma z \prec_\sigma x$?
    \end{description}
    \label{prob:ordering2}
\end{problem}

\noindent Besides the tuples in $B$, in both problems the second and the third entry of the triples in $C$ imply simple order constraints on the elements of $A$.  Therefore, we can define the relations $\pi(B,C)$ and $\pi(C)$, respectively, as the reflexive and transitive closure of the relation $R \subseteq A \times A$ where $(y,z) \in R$ if and only if $(y,z) \in B$ or there is some tuple $(x,y,z) \in C$. If $(A,C)$ is a positive instance of the General Intermezzo problem, then $\pi(C)$ must form a partial order and every intermezzo ordering of $(A,C)$ forms a linear extension of $\pi(C)$. The same properties hold for positive instances $(A,B,C)$ of the Intermezzo problem and the partial order $\pi(B,C)$. Thus, we also can interpret the (General) Intermezzo problem as a special kind of linear extension problems with additional non-betweenness constraints. This motivates the consideration of restricted problems where the partial order has to fulfill certain properties.

\subsection{The Intermezzo Problem for CS-trees}

Following the terminology of Trotter~\cite{trotter1992combinatorics}, we call a partial order a \emph{cs-tree} (short for \emph{computer science tree}) if its Hasse diagram forms a tree rooted in the unique minimal element. Using the terminology in \cite{scheffler2023graph}, we call a leaf of a rooted tree a \emph{branch leaf} if it is not equal to the root of the tree. Recall that the height of a cs-tree and the height of the tree that is formed by its Hasse diagram differ by one.
 
\begin{lemma}\label{lemma:tree-intermezzo}~
\begin{enumerate}
    \item The rooted \cl-tree recognition problem of Generic Search for rooted spanning trees of fixed height~$h \geq 2$ is polynomial-time reducible to the General Intermezzo problem for instances $(A,C)$ where $\pi(C)$ is a cs-tree of height $h+1$.
    \item The General Intermezzo problem for instances $(A,C)$ where $\pi(C)$ is a cs-tree of fixed height $h \geq 2$ is polynomial-time reducible to the rooted \cl-tree recognition problem of Generic Search for rooted spanning trees of height~$2h - 1$.
    \item The rooted \cl-tree recognition problem of Generic Search for rooted spanning trees having $k$ branch leaves is polynomial-time equivalent to the General Intermezzo problem for instances $(A,C)$ where $\pi(C)$ is a cs-tree of width $k$.
\end{enumerate}
\end{lemma}
\begin{proof}
    First, we reduce the \cl-tree recognition problem of GS to the General Intermezzo Problem. Let $G$ be a graph and $T$ be a spanning tree of $G$ rooted in $r$ of height~$h \geq 2$. Let $A = V(G) \cup \{s\}$ where $s \notin V(G)$. Let $C$ be the set containing the following triples: 
    
    \begin{enumerate}[(C1)]
        \item $(r, t, s)$, for some child $t$ of $r$ in $T$,\label{item:c1}
        \item $(s, u, v)$, for any vertex $v \in V(G) \setminus \{r\}$ and its parent $u$ in $T$,\label{item:c2}
        \item $(w, u, v)$, for any vertex $v \in V(G)  \setminus \{r\}$, its parent $u$ in $T$ and any vertex $w$ with $vw \in E(G) \setminus E(T)$.\label{item:c3}
    \end{enumerate} 
    
    It is easy to see that $\pi(C)$ is a cs-tree of height~$h + 1$.  We claim that there is an intermezzo ordering of $A$ fulfilling the constraints given by $C$ if and only if $T$ is a rooted GS \cl-tree of $G$.

    First assume that there is an intermezzo ordering $\sigma$ fulfilling the constraints of $C$. We delete $s$ from $\sigma$ and call the resulting ordering $\sigma'$. The following claim is implied by the constraints given in (C\ref{item:c2}).

    \begin{claim1}\label{claim:1}
        If $u$ is the parent of $v$ in $T$, then $u \prec_{\sigma'} v$.
    \end{claim1}

    This claim implies directly that $\sigma'$ is a GS ordering of $G$ starting in $r$. Now assume for contradiction that the \cl-tree $T'$ of $\sigma'$ is not equal to $T$. Then there is a vertex $v$ whose parent $u'$ in $T'$ is different from its parent $u$ in $T$. By \cref{claim:1}, it holds that $u \prec_{\sigma'} v$. Therefore, it must hold that $u \prec_{\sigma'} u' \prec_{\sigma'} v$. This implies that $u'$ is not a child of $v$ in $T$, due to \cref{claim:1}. Hence, the edge $u'v$ is not part of $T$ but part of $G$. Then, the set $C$ contains the triple $(u',u,v)$ (see (C\ref{item:c3})). This is a contradiction because $\sigma'$ and, thus, $\sigma$ does not fulfill the constraint given by that triple.  

    Now assume that $T$ is the \cl-tree of the GS ordering $\sigma$ starting with $r$. Then let $\sigma'$ be the ordering constructed by appending $s$ to the end of $\sigma$. The ordering fulfills the constraint $(r,t,s)$ given in (C\ref{item:c1}). Furthermore, as parents are to the left of their children in $\sigma$, the ordering $\sigma'$ also fulfills the constraints given by (C\ref{item:c2}). Assume for contradiction that some triple $(w,u,v)$ of (C\ref{item:c3}) is not fulfilled in $\sigma'$, i.e., $u \prec_{\sigma'} w \prec_{\sigma'} v$. Then $u$ is not the parent of $v$ in the \cl-tree of $\sigma$ since $w$ is a neighbor of $v$ and $w$ lies between $u$ and $v$ in $\sigma'$; a contradiction.

    This proves the first statement of the lemma. To prove the same direction for the third statement, we slightly change the set $C$. Instead of the triple $(r, t, s)$ given in (C\ref{item:c1}), we add the triple $(t,s,r)$. It is easy to see that then the width of $\pi(C)$ is equal to the number of branch leaves of $T$. The rest of the proof works analogously with the only difference that we append $s$ to the beginning and not to the end of the GS ordering of $G$.

    Now we reduce the General Intermezzo problem where $\pi(C)$ is a cs-tree to the rooted \cl-tree recognition problem of GS. Let $(A,C)$ be an instance where $\pi(C)$ is a cs-tree of height $h$ and width $k$. We define the vertex set to be $V(G) := \{v^1, v^2~|~v \in A\}$. Let $H$ be the Hasse diagram of~$\pi(C)$. The set of edges of $T$ is defined as $E(T) := \{v^1v^2~|~v \in A\} \cup \{u^2v^1~|~uv \in E(H) \wedge (u,v) \in \pi(C)\}$. Let $(x,y,z) \in C$ and let $y = w_0, w_1, \dots, w_\ell = z$ be the elements of the path between $y$ and $z$ in the Hasse diagram of~$\pi(C)$. We add a non-tree edge to $G$ from $x^2$ to any vertex $w^1_i$ and $w^2_i$ with $i \geq 1$. It is obvious that the constructed tree $T$ has height $2h-1$ and $k$ branch leaves if it is rooted in the minimal element of~$\pi(C)$.

    First assume that there is an intermezzo ordering $\sigma$ for $(A,C)$. Then let $\sigma'$ be the ordering that is constructed by replacing every element $v \in A$ in $\sigma$ by the ordering $(v^1, v^2)$. Then we claim that $\sigma'$ is a GS ordering of $G$ having \cl-tree $T$. Let $T'$ be the \cl-tree of $\sigma'$. Obviously, it holds for $\sigma'$ that any vertex is to the right of its parent in $T$ since $\sigma'$ is constructed from a linear extension of $\pi(C)$. Therefore, $\sigma'$ is a GS ordering of $G$. Furthermore, any vertex $v^2$ has parent $v^1$ both in $T$ and $T'$ since these two vertices are consecutive in $\sigma'$. Assume for contradiction that there is a vertex $v^1$ whose parent in $T$ is $u^2$ but the parent of $v^1$ in $T'$ is $t^p$ with $u^2 \neq t^p$. By construction it holds that $t^p = t^2$ since $v^1$ has no neighbor with index~1. It holds that $u^2 \prec_{\sigma'} t^2 \prec_{\sigma'} v^1$ and $t^2v^1$ is an edge in $E(G) \setminus E(T)$. This non-tree edge has been added to $E(G)$ because of some triple $(t, a, b) \in C$ where $a$ is an ancestor of $v$ and $b$ is a descendant of $v$ in the Hasse diagram of $\pi(C)$ (or $b = v$). However, by construction of $\sigma'$, it holds that $a^2 \prec_{\sigma'} u^2 \prec_{\sigma'} t^2 \prec_{\sigma'} v^1 \prec_{\sigma'} b^2$. Hence $a \prec_\sigma t \prec_\sigma b$; a contradiction to the fact that $\sigma$ is an intermezzo ordering of $C$.

    Now assume that there is a GS ordering $\sigma$ of $G$ having $T$ as its \cl-tree.  We construct the ordering $\sigma'$ of $A$ as follows. Consider the subordering of $\sigma$ containing only the vertices $v^2$ for any $v \in A$ and replace $v^2$ by $v$. We claim that $\sigma'$ is an intermezzo ordering of $(A,C)$. Let $(x,y,z)$ be a triple in $C$. It holds that $y \prec_\sigma z$ since $y^2$ is an ancestor of $z^2$ in $T$. Assume for contradiction that $y \prec_{\sigma'} x \prec_{\sigma'} z$. Consider the $w_0, \dots, w_\ell$ as defined above. Note that these vertices appear by ascending index in $\sigma'$. Let $w_i$ be the leftmost of these vertices in $\sigma'$ that is to the right of $x$. Then it holds that $y \preceq_{\sigma'} w_{i-1} \prec_{\sigma'} x \prec_{\sigma'} w_i \preceq_{\sigma'} z$. This implies that $w^2_{i-1} \prec_{\sigma} x^2 \prec_{\sigma} w^2_i$. Now we consider two cases. If $w^1_i \prec_\sigma x^2$, then it must hold that $w^2_{i-1} \prec_{\sigma} w^1_i \prec_{\sigma} x^2 \prec_{\sigma} w^2_i$ as $w^2_{i-1}$ is the parent of $w^1_i$ in $T$. However, by construction $x^2$ is adjacent to $w^2_i$ and, hence, $w^1_i$ cannot be the parent of $w^2_i$ in the \cl-tree of $\sigma$. If $x_2 \prec_\sigma w^1_i$, then it must hold that $w^2_{i-1} \prec_{\sigma} x^2 \prec_{\sigma} w^1_i \prec_{\sigma} w^2_i$ as $w^1_i$ is the parent of $w^2_i$ in $T$. However, by construction $x^2$ is adjacent to $w^1_i$ and, hence, $w^2_{i-1}$ cannot be the parent of $w^1_i$ in the \cl-tree of $\sigma$. This contradicts the choice of $\sigma$.       
\end{proof}

\noindent Combining the first statement of the lemma with \cref{thm:main} yields the following result.

\begin{theorem}\label{thm:general-np}
    The General Intermezzo problem is \NP-complete even if the input is restricted to instances $(A,C)$ where $\pi(C)$ is a cs-tree of height~$6$.
\end{theorem}
We can extend this result to the Intermezzo problem as follows.

\begin{lemma}\label{lemma:general-to-normal}
    The General Intermezzo Problem for instances $(A,C)$ where $\pi(C)$ is a cs-tree of width $k$ is polynomial-time reducible to the Intermezzo Problem for instances $(A',B',C')$ where $\pi(B',C')$ is a cs-tree of width $k$.
\end{lemma}
\fullversion{\begin{proof}
    Let $(A,C)$ be an instance of the General Intermezzo Problem where $\pi(C)$ is a cs-tree of width $k$. Let $T$ be the rooted tree that corresponds to the Hasse diagram of $\pi(C)$. In the following we will construct an equivalent instance $(A', B', C')$ of the intermezzo problem.

    We construct this instance in three steps. First, we ensure that any element $x \in A'$ that appears as the first element of a tuple in $C'$ occurs only in one tuple. Afterwards, we will do the same for the second and the third element. 

    Let $A = \{v_1, \dots, v_n\}$. For any $v_i$ let $\gamma^i_1, \dots, \gamma^i_{m_i}$ be the triples in $C$ containing $v_i$ as first element. We add the following elements to $A'$:

    \begin{enumerate}
        \item $a^i$ and $d^i$ for any $1 \leq i \leq n$,
        \item $b^i_{j}$ for any $1 \leq i,j \leq n$,
        \item $c^i_j$ for any $1 \leq i \leq n$, $1 \leq j \leq m_i$.
    \end{enumerate}

    We order these elements in $B'$ linearly as follows: $a^i \prec_{B'} b^i_{1} \prec_{B'} b^i_{2} \prec_{B'} \dots \prec_{B'} b^i_{n} \prec_{B'} c^i_1 \prec_{B'} \dots \prec_{B'} c^i_{m_i} \prec_{B'} d^i$. Furthermore, we add $(d^i,a^j)$ to $B'$ if $(v_i,v_j) \in \pi(C)$.
    Note that, so far, the partial order $\pi(B', C')$ is still a cs-tree of width $k$.

    Next we consider the following set of triples.
    \[
    C'_1 = \{(b^i_{j}, a^j, d^j)~|~1 \leq i,j \leq n, i \neq j\}.\] 
    
    Let $\sigma$ be an intermezzo ordering of $A'$ fulfilling the constraints of $B'$ and those of $C'_1$.

    \begin{claim1}\label{claim:gtn1}
         If $a^j \prec_\sigma a^i$, then $c^j_\ell \prec_\sigma d^i$ for all $\ell \in \{1,\dots, m_j\}$.
     \end{claim1}

     \begin{claimproof}
         Assume for contradiction that $d^i \prec_\sigma c^j_\ell$ for some $\ell \in \{1,\dots,m\}$. Due to the constraints of $B'$, it holds that $a^i \prec_\sigma b^i_j \prec_\sigma d^i$ and $c^j_\ell \prec_\sigma d^j$. Combining this with $a^j \prec_\sigma a^i$, we get that $a^j \prec_\sigma b^i_j \prec_\sigma d^j$. However, this implies that $\sigma$ does not fulfill the triple $(b^i_{j}, a^j, d^j) \in C'_1$; a contradiction.
     \end{claimproof}

     Next we consider the following set of triples.
     \[
        C'_2 := \{(c^i_\ell, a^h, d^j)~|~\gamma^i_\ell = (v_i, v_h, v_j)\}
     \]

    \begin{claim1}
        There is an intermezzo ordering of $A$ fulfilling the constraints given by $B$ and $C$ if and only if there is an intermezzo ordering of $A'$ fulfilling the constraints given by $B'$, $C'_1$, and $C'_2$.
    \end{claim1}

    \begin{claimproof}
        Let $\sigma$ be an intermezzo ordering of $A$ fulfilling the constraints given by $B$ and $C$. We replace any $v_i$ in $\sigma$ by the ordering $a^i, b^i_1, \dots, b^i_n, c^i_1, \dots, c^i_{m_i}, d^i$ and get the ordering $\sigma'$ of $A'$. It is obvious that $\sigma'$ fulfills the constraints given by $B'$. It fulfills the constraints of $C'_1$ since for any element $b^i_j$ between $a^j$ and $d^j$ it holds that $i = j$. Furthermore, if for some triple $(c^i_\ell, a^h, d^j) \in C'_2$ it holds that $a^h \prec_{\sigma'} c^i_\ell \prec_{\sigma'} d^j$, then $v_h \prec_\sigma v_i \prec_\sigma v_j$. However, this is not possible as there is a triple $(v_i, v_h, v_j) \in C$ which is fulfilled by $\sigma$.

        Now let $\sigma'$ be an intermezzo ordering of $A'$ fulfilling the constraints given by $B'$, $C'_1$ and $C'_2$. Consider the subordering of $\sigma'$ that only contains all the elements $a^i$ for any $i \in \{1,\dots,n\}$. We replace any $a^i$ in that ordering by $v_i$ and get the ordering $\sigma$ of $A$. This ordering fulfills the constraints given by $B$. Now assume for contradiction that $\sigma$ does not fulfill the constraint $\gamma^i_\ell = (v_i, v_h, v_j)$, i.e., $v_h \prec_\sigma v_i \prec_\sigma v_j$. By construction of $\sigma$, it holds that $a^h \prec_{\sigma'} a^i \prec_{\sigma'} a^j$. Due to the constraints in $B'$ we know that $a^i \prec_{\sigma'} c^i_\ell$. Furthermore, by \cref{claim:gtn1}, it holds that $c^i_\ell \prec d^j$. Hence, it holds that $a^h \prec_{\sigma'} c^i_\ell \prec_{\sigma'} d^j$. However, this implies that $\sigma'$ does not fulfill the constraint $(c^i_\ell, a^h, d^j)$; a contradiction.
    \end{claimproof}

    Let $C' = C'_1 \cup C'_2$. Note that any $b^i_j$ and any $c^i_h$ appears in at most one triple of $C'$. Furthermore, the elements $a^i$ and $d^i$ only appear as second or third element of triples, respectively. To make the instance $(A', B', C')$ an instance of the intermezzo problem, we have to ensure that $a^i$ and $d^i$ also appear at most once in a triple of $C'$.

    Let $a^i$ be an arbitrary element and let $(x_1, a^i, z_1), \dots, (x_\ell, a^i, z_\ell)$ be the triples in $C'$ where $a^i$ appears in. Furthermore, let $d^j$ be the parent of $a^i$ in the Hasse diagram of $\pi(B',C')$ if there is one. We add the new elements $a^i_1, \dots, a^i_\ell$ to $A'$. Furthermore, we add the following tuples to $B$: $(a^i_t, a^{i}_{t+1})$ for any $t \in \{1,\dots, \ell - 1\}$, $(a^i_\ell, a)$, and $(d^j, a^i_1)$. For any $t \in \{1, \dots, \ell\}$, we replace $(x_t, a^i, z_t)$ in $C'$ by $(x_t, a_t^i, z_t)$. We repeat this procedure for every $a^i \in A'$. Let $(A'',B'',C'')$ be the resulting instance. Note that any element of $A''$ that appears as second element in some triple of $C''$ does not appear in any other triple of $C''$. Furthermore, $\pi(B'',C'')$ is still a cs-tree of width $k$.
    
    \begin{claim1}\label{claim:general-to-normal1}
        There is an intermezzo ordering of $A'$ fulfilling the constraints given by $B'$ and $C'$ if and only if there is an intermezzo ordering of $A''$ fulfilling the constraints given by $B''$ and $C''$.
    \end{claim1}

    \begin{claimproof}
    Let $\sigma''$ be an intermezzo ordering for $(A'', B'', C'')$. Let $\sigma'$ be the subordering of $\sigma''$ restricted to $A'$. It is obvious that $\sigma'$ fulfills the constraints of $B'$. Assume for contradiction that some original triple $(x_t, a^i, z_t)$ is not fulfilled, i.e.,  $a^i \prec_{\sigma'} x_t \prec_\sigma z_t$. As $(a^i_t, a^i)$ is an element of $B''$, we know that $a^i_t \prec_{\sigma''} a^i$. However, this is a contradiction to the choice of $\sigma''$ since $(x_t, a^i_t, z_t)$ is an element of $C''$.

    For the other direction let $\sigma'$ be an intermezzo ordering of $(A', B', C')$. For any $a^i \in A'$, we add the elements $a^i_1, \dots, a^i_\ell$ directly before $a^i$ in $\sigma'$. We call the resulting ordering $\sigma''$. This ordering fulfills the constraints given by $B''$. Assume for contradiction that there is some triple $(x_t, a^i_t, z_t)$ in $C''$ that is not fulfilled by $\sigma''$, i.e., $a^i_t \prec_{\sigma''} x_t \prec_{\sigma'} z_t$. Then it holds that $a^i \prec_{\sigma'} x_t \prec_{\sigma''} z_t$; a contradiction to choice of $\sigma'$ since $(x_t, a^i, z_t)$ is in $C'$.
    \end{claimproof}

    To ensure that any $d^j$ appears in at most one triple, we can use the same construction. The only difference is that the new elements $d^j_1, \dots, d^j_\ell$ are made descendants of $d^j$ in the Hasse diagram of $\pi(B'',C'')$. The proof of correctness works like the proof of \cref{claim:general-to-normal1}.
\end{proof}}
This lemma implies the following complexity result.

\begin{theorem}\label{thm:intermezzo-np}
    The Intermezzo problem is \NP-complete even if the input is restricted to instances $(A,B,C)$ where $\pi(B,C)$ is a cs-tree.
\end{theorem}
As was shown by Wolk~\cite{wolk1965comparability}, cs-trees have dimension~2. Combining this with \cref{thm:general-np,thm:intermezzo-np}, we get the following result.

\begin{corollary}
    The (General) Intermezzo problem is \NP-complete even if $\pi(B,C)$ or $\pi(C)$, respectively, has dimension at most 2. 
\end{corollary}
Note that -- in difference to \cref{thm:general-np} -- we were not able to bound the height of the partial order in \cref{thm:intermezzo-np} since in the proof of \cref{lemma:general-to-normal} the height of the constructed partial order depends on the number of elements in $A$. We can adapt the proof of that lemma such that the height of the partial order increases only by a constant factor. However, in this case, we then loose the property that the partial order is a cs-tree.

\begin{corollary}\label{corr:36}
    The Intermezzo problem is \NP-complete even if the input is restricted to instances $(A,B,C)$ where $\pi(B,C)$ has height~$36$.
\end{corollary}
\fullversion{\begin{proof}
    We adapt the proof of \cref{lemma:general-to-normal}. In the first part of the proof it is only necessary that $a^i$ is before any $b^i_j$, any $b^i_j$ is before any $c^i_j$ and any $c^i_j$ is before $d^i$. In the second and the third part we only have to ensure that any $a^i_t$ is before $a^i$ and that any $d^j_t$ is after $d^j$. Adding all these constraints to $B$ only increases the height of the partial order by factor~6 since any element of $A$ is replaced by parallel chains of size~6. It follows from \cref{thm:general-np} that we can bound the height by~36. 
\end{proof}}
The complexity of the Intermezzo problem for cs-trees of bounded height remains open.

\subsection{The Intermezzo Problem for Partial Orders of Bounded Width}

As we have seen in the section above, the (General) Intermezzo problem is \NP-complete even if the height or the dimension of the partial order is bounded. One may ask whether this also holds for another notable parameter of partial orders, the width. Adapting an idea of Colbourn and Pulleyblank~\cite{colbourn1985minimizing} (explained in more detail in~\cite{beisegel202?computing}), we can show that -- unless $\P = \NP$ -- this is not the case as we can give an $\XP$-algorithm for the General Intermezzo problem parameterized by the width of $\pi(C)$.

\begin{theorem}\label{thm:xp-intermezzo}
    The General Intermezzo problem can be solved in time $\O(k\cdot n^{k+2})$ on any instance $(A,C)$ where $n = |A|$ and $k$ is the width of $\pi(C)$.
\end{theorem}
\begin{proof}
    We only sketch the idea of the algorithm; for a comprehensive description and analysis of a similar algorithm see~\cite{beisegel202?computing}. Using Dilworth's Chain Covering Theorem~\cite{dilworth1987decomposition}, we can partition the set $A$ into $k$ disjoint chains of $\pi(C)$. Now the set of elements of any prefix $\sigma^\text{pre}$ of a linear extension of $\pi(C)$ can be represented by a tuple $(a_1, \dots, a_k) \in \{0,\dots,|A|\}^k$ where $a_i$ represents the number of elements of chain $i$ that are part of $\sigma^\text{pre}$. Since all elements of a chain are strictly ordered, the number of used elements of the chain directly implies the elements of the chain that are part of the prefix set. 

    Now the algorithm uses dynamic programming to compute whether a given prefix set can be reached in such a way that it fulfills all conditions of $C$. To this end, we have a table $M$ with 0-1-entries for every tuple representing a prefix. We fill the entries of this table inductively, starting with those tuples whose entries sum up to 1. Such a tuple gets a 1-entry in $M$ if and only if the minimal element of the respective chain is a minimal element of the partial order. For tuples $\gamma = (a_1, \dots, a_k)$ with larger entry sums we check for any tuple $\gamma'$ that is constructed by decrementing exactly one non-zero entry of $\gamma$, say $a_i$, the following:
    \begin{enumerate}
        \item Is the $M$-entry of tuple $\gamma'$ equal to 1?
        \item Is the $a_i$-th element $x$ of the $i$-th chain minimal in $\pi(C)$ restricted to those elements that are not part of the prefix set encoded by $\gamma'$?
        \item Is there no triple $(x,y,z)$ with $y$ is an element of the prefix set encoded by $\gamma'$ and $z$ is not such an element?
    \end{enumerate}
    If the answer to all three question is yes, then we set the $M$-entry of $\gamma$ to 1. It is easy to check that $M$ has $\O(n^k)$ entries and for any entry we can answer the three questions above for all triples $\gamma'$ in time $\O(kn^2)$. This leads to the claimed running time.
\end{proof}
\Cref{lemma:tree-intermezzo} implies an \XP-algorithm for the rooted \cl-tree recognition of GS parameterized by the number of branch leaves of the given spanning tree.

\begin{corollary}
    The rooted \cl-tree recognition problem of Generic Search can be solved in time $\O(k\cdot n^{k+2})$ on a graph $G$ and a rooted spanning tree $T$ of $G$ where $n = |V(G)|$ and $k$ is the number of branch leaves of $T$.
\end{corollary}
One may ask whether this rather poor running time bound given in \cref{thm:xp-intermezzo} can be improved significantly and whether there is an \FPT{} algorithm for the (General) Intermezzo problem parameterized by the width of the partial order. We will show that -- under certain assumptions~-- this is not the case.

\begin{theorem}\label{thm:w1}
    The (General) Intermezzo problem is \W-hard if it is parameterized by the width $k$ of $\pi(C)$ or $\pi(B,C)$, respectively, even if that partial order is a cs-tree. Furthermore -- under the assumption of the Exponential Time Hypothesis -- there is no algorithm that solves the problem in time $f(k) \cdot n^{o(k)}$ for any computable function $f$ where $n = |A|$.
\end{theorem}
We prove this result by an \FPT-reduction from the following problem, applying a technique also used in~\cite{beisegel202?computing}.

\begin{problem}[Multicolored Clique Problem (MCP)]~
\begin{description}
\item[\textbf{Instance:}] A graph $G$ with a proper coloring of $k$ colors.
\item[\textbf{Question:}]
Is there a clique in $G$ that contains exactly one vertex of each color?
 \end{description}
\end{problem}

\noindent The MCP was shown to be \W-hard by Pietrzak~\cite{zbMATH02092872} and independently by Fellows et~al.~\cite{fellows2009param}. In fact, in \cite{cygan2015param, lokshtanov2011lower} the authors show the following result.
\begin{theorem}[Cygan et al.~\cite{cygan2015param}, Lokshtanov et al.~\cite{lokshtanov2011lower}]\label{thm:lower}
Under the assumption of the Exponential Time Hypothesis, there is no $f(k) n^{o(k)}$ time algorithm for the Multicolored Clique Problem for any computable function $f$ where $n$ is the number of vertices of the given graph.
\end{theorem}
We give an \FPT-reduction from the MCP to the General Intermezzo problem parameterized by the width of $\pi(C)$. \Cref{lemma:general-to-normal} implies such a reduction also for the Intermezzo problem.

Let $G$ be an instance of the MCP with $k$ colors. W.l.o.g. we may assume that every color class has exactly $q$ elements, i.e., we assume that $V(G) = \{v^i_p~|~1 \leq i \leq k, 1 \leq p \leq q\}$. In the following, we construct an equivalent instance $(A,C)$ for the General Intermezzo problem.

First we describe the set $A$. For every $i \in \{1,\dots, k\}$ and every $p \in \{1,\dots,q\}$, we define the set $U^i_p := \{u^i_{p,j}~|~0 \leq j \leq k\}$. The set $U^i$ is defined as $U^i := \bigcup_{1 \leq p \leq q} U^i_p$. Now set $A$ is defined as follows.
\[
A:= \{s^i~|~1 \leq i \leq k+1\} \cup \{c^{i,j}~|~1 \leq i \leq j \leq k\} \cup \bigcup_{1 \leq i \leq k} U^i.
\]
In the remainder of the section, we construct the set $C$ by adding subsets of triples with specific properties. We start with some simple order constraints. For the sake of convenience, we only give a set $B$ of tuples encoding these constraints. Note that these tuples can also be encoded using triples by introducing a new element that is not allowed to be between any of the elements of those tuples.%
\begin{alignat*}{4}
B := &\{(s^i, u^i_{p,j})~&&|~1 \leq i \leq k,~0 \leq j \leq k,~1 \leq p \leq q\}~\cup \\
     &\{(u^i_{p,j}, u^i_{r,\ell})~&&|~1 \leq i \leq k,~p < r \text{ or } p = r \text{ and } j < \ell\}~\cup \\
     &\{(u^i_{1,0}, s^{i+1})~&&|~1 \leq i \leq k\}~\cup \\
     &\{(c^{i,j}, c^{\ell,m})~&&|~i < \ell \text{ or } i = \ell \text{ and } j < m\}~\cup \\
     &\{(c^{i,k}, c^{i+1,i+1})~&&|~i < k\}~\cup \\
     &\{(s^{k+1}, c^{1,1})\} &&%
\end{alignat*}

\begin{lemma}\label{lemma:width}
The reflexive and transitive closure of $B$ forms a cs-tree of width $k+1$.
\end{lemma}
\fullversion{\begin{proof}
    The following set forms an antichain of size $k+1$: $\{u^i_{1,1}~|~1 \leq i \leq k\} \cup \{s^{k+1}\}$. Thus, the width is at least $k+1$. 
    
    We can cover the set $A$ by $k+1$ chains in the following way: The set $\{s^i, u^i_{1,0}~|~1\leq i \leq k\} \cup \{s^{k+1}\} \cup \{c^{i,j}~|1\leq i \leq j \leq k\}$ forms a chain. Furthermore, for every $i \in \{1,\dots,k\}$, the set $U^i$ forms a chain. As every antichain contains at most one element of every of those $k+1$ chains, the width is at most $k+1$. It is easy to see that the union of these $k+1$ chains forms the Hasse diagram of the reflexive and transitive closure of $B$. Therefore, this closure is a cs-tree.
\end{proof}}
In the following, we will present the rest of the triples of the set $C$. First note that the last two elements of these triples will also be contained as a tuple in the reflexive and transitive closure of $B$. Thus, they do not contribute any new tuples to $\pi(C)$ and, hence, \cref{lemma:width} implies that $\pi(C)$ is a cs-tree of width $k+1$. 

We will present the new triples not all at once. Instead we present specific subsets of them. Then we will give properties that are fulfilled by any intermezzo ordering of $A$ that fulfills the constraints of $B$ and all triples presented up to that point. In any of these properties, the ordering $\sigma$ will be the respective intermezzo ordering of $A$. We divide this ordering into a \emph{selection phase} where we choose one vertex of every color to be part of the candidate clique. In the \emph{verification phase}, we check whether the chosen vertices indeed form a clique in $G$. We start with the triples for the selection phase.
\begin{alignat*}{4}
    C^1_\text{sel} := &\{(s^{i+1}, u^i_{p,j}, u^i_{p+1, 0})~&&|~1 \leq i \leq k, 1 \leq p < q, 1 \leq j \leq k\} \\
    C^2_\text{sel} := &\{(u^i_{q,j}, s^{i}, s^{i+1})~&&|~1 \leq i \leq k, 1 \leq j \leq k\} \\
    C^3_\text{sel} := &\{(u^i_{p,j}, s^{i+1}, c^{1,1})~&&|~1 \leq i \leq k, 1 \leq p \leq q, 0 \leq j \leq k\}
\end{alignat*}

\begin{property}\label{property:sel}
There exist indices $p_1, \ldots, p_k \in \{1,\ldots,q\}$ such that the prefix $\sigma'$ of $\sigma$ ending in $c^{1,1}$ fulfills the following conditions: 
\begin{enumerate}
  \item $\sigma'$ starts with $s^1$ and does not contain any vertices $c^{i,j}$ with $i \neq 1$ or $j \neq 1$.
  \item for all $i \in \{1,\dots,k\}$ it holds:\label{item:lem3-0}
  \begin{enumerate}
    \item vertex $s^i$ and $\bigcup_{r=1}^{p_i-1}U^i_{r}$ are part of $\sigma'$, \label{item:lem3-1}
    \item $U^i_{p_i} \cap \sigma' = \{u^{i}_{p_i,0}\}$,\label{item:lem3-2}
    \item none of the vertices of $U^i_{r}$ with $r > p_i$ are part of $\sigma'$.\label{item:lem3-3}
  \end{enumerate}
\end{enumerate}
\end{property}
\fullversion{\begin{proof}
    It follows from the conditions of $B$ that all the elements $s^i$ have to be to the left of $c^{1,1}$ in $\sigma$ and, thus, they are contained in $\sigma'$. $B$ also implies that there is at least one element of $U^i$ to the left of $s^{i+1}$. The triples of $C^2_\text{sel}$ imply that not all elements of $U^i$ can be to the left of $s^{i+1}$. Let $u^i_{p_i,j}$ be the rightmost element of $U^i$ to the left of $s^{i+1}$. If $p_i = q$, then $C^2_\text{sel}$ implies that $j = 0$ since none of the elements $u^i_{q,\ell}$ with $\ell > 0$ can be between $s^i$ and $s^{i+1}$. If $p_i < q$ and $j \neq 0$, then this would contradict the triples of $C^1_\text{sel}$ since $s^{i+1}$ would lie between $u^i_{p_i,j}$ and $u^i_{p_i + 1, 0}$. Thus, in any case $j = 0$. Therefore, to the left of $s^{i+1}$, all three conditions of \ref{item:lem3-0} are fulfilled. The triples of $C^3_\text{sel}$ ensure that no element of $U^i$ can be taken between $s^{i+1}$ and $c^{1,1}$. Therefore, the conditions of \ref{item:lem3-0} also hold for $\sigma'$.
\end{proof}}

\noindent We now present the first triples for the verification phase. They ensure that between certain $c$-elements only some elements are allowed to be taken. By~\cref{property:sel} we assume in the following that $p_1, \ldots, p_k$ are fixed. 
\begin{alignat*}{4}
C^1_\text{ver} &:= \{(x,c^{i,k},c^{i+1,i+1})~&&|~1 \leq i < k,~x \in A \setminus \{c^{i,k}, c^{i+1,i+1}\}\} \\
C^2_\text{ver} &:= \{(u^i_{p,j}, c^{\ell, m}, c^{\ell, m+1})~&&|~i \neq \ell \text{ and } i \neq m+1 \text{ or } i= m + 1 \text{ and } j \neq \ell \text{ or } \\&&&\phantom{|~}i = \ell \text{ and } j \neq m\} 
\end{alignat*}

\begin{property}\label{property:ver1}
   There are at most two elements in $\sigma$ between $c^{\ell, m}$ and $c^{\ell, m+1}$, namely $u^\ell_{p_\ell,m}$ and $u^{m+1}_{p_{m+1},\ell}$.
\end{property}
\fullversion{\begin{proof}
    We first observe that no element $u^i_{r,0}$ can be between $c^{1,1}$ and $c^{k,k}$ as $C^1_\text{ver}$ forbids them to be between $c^{i,k}$ and $c^{i+1,i+1}$. Furthermore, $C^2_\text{ver}$ forbids them to be between any $c^{\ell, m}$ and $c^{\ell, m+1}$ (as there is no $c^{\ell,0}$). 
    By \cref{property:sel}, for any set $U^i_r$ with $r \neq p_i$ either all elements are to the left of $c^{1,1}$ or the element $u^i_{r,0}$ is not to the left of $c^{1,1}$. Combining this with the observation above and the conditions of $B$, we know that none of the elements of $U^i_r$ with $r \neq p_i$ are between $c^{1,1}$ and $c^{k,k}$. Now, the triples of $C^2_\text{ver}$ directly imply that between $c^{\ell, m}$ and $c^{\ell, m+1}$ there can only be the elements $u^\ell_{p_\ell,m}$ and $u^{m+1}_{p_{m+1},\ell}$.
\end{proof}}
In the next step we want to ensure that $u^\ell_{p_\ell,m}$ and $u^{m+1}_{p_{m+1},\ell}$ have to be taken between $c^{\ell, m}$ and $c^{\ell, m+1}$.
\begin{alignat*}{4}
C^3_\text{ver} := &\{(c^{\ell, m +1}, u^\ell_{p,m-1}, u^\ell_{p,m})~&&|~1 \leq \ell \leq m < k, 1\leq p \leq q\}~\cup \\
             &\{(c^{\ell, m +1}, u^{m+1}_{p,\ell-1}, u^{m+1}_{p,\ell})~&&|~1 \leq \ell \leq m < k, 1\leq p \leq q\} 
\end{alignat*}

\begin{property}\label{property:ver2}
    The elements $u^\ell_{p_\ell,m}$ and $u^{m+1}_{p_{m+1},\ell}$ have to be between $c^{\ell, m}$ and $c^{\ell, m+1}$ in $\sigma$.
\end{property}
\fullversion{\begin{proof}
    We prove the property inductively. Assume that $u^\ell_{p_\ell,m-1}$ and $u^{m+1}_{p_{m+1},\ell-1}$ are to the left of $c^{\ell,m}$ and, thus, to the left of $c^{\ell,m+1}$ in $\sigma$. Note that this holds if $m = 1$ or $\ell = 1$, due to \cref{property:sel}. By \cref{property:sel,property:ver1}, 
    $u^\ell_{p_\ell,m}$ and $u^{m+1}_{p_{m+1},\ell}$ cannot be to the left of $c^{\ell,m}$ in $\sigma$. Therefore, the triples of $C^3_\text{ver}$ imply that both $u^\ell_{p_\ell,m}$ and $u^{m+1}_{p_{m+1},\ell}$ have to be taken before $c^{\ell, m+1}$.
\end{proof}}
Finally, we have to ensure that $u^\ell_{p_\ell,m}$ and $u^{m+1}_{p_{m+1},\ell}$ can only be taken if $v_p^\ell v_{p_{m+1}}^{m+1} \in E(G)$. This is ensured by the following triples.
\begin{alignat*}{4}
C^4_\text{ver} := &\{(u^{m+1}_{r,\ell}, u^\ell_{p,m}, u^\ell_{p,m+1})~&&|~1 \leq \ell \leq m < k, 1\leq p \leq q, 1 \leq r \leq q, v_p^\ell v_r^{m+1} \notin E(G)\}~\cup \\
             &\{(u^{\ell}_{p,m}, u^{m+1}_{r,\ell}, u^{m+1}_{r,\ell+1})~&&|~1 \leq \ell \leq m < k, 1\leq p \leq q, 1 \leq r \leq q, v_p^\ell v_r^{m+1} \notin E(G)\}
\end{alignat*}

\begin{property}\label{property:ver3}
   The elements $u^\ell_{p_\ell,m}$ and $u^{m+1}_{p_{m+1},\ell}$ can be between $c^{\ell, m}$ and $c^{\ell, m+1}$ in $\sigma$ only if $v_{p_\ell}^\ell v_{p_{m+1}}^{m+1} \in E(G)$.
\end{property}
\fullversion{\begin{proof}
    Assume $u^\ell_{p_\ell,m}$ and $u^{m+1}_{p_{m+1},\ell}$ are between $c^{\ell, m}$ and $c^{\ell, m+1}$ but $v_{p_\ell}^\ell v_{p_{m+1}}^{m+1} \notin E(G)$. Due to \cref{property:ver1}, no other elements are between $c^{\ell, m}$ and $c^{\ell, m+1}$. If $u^\ell_{p_\ell,m} \prec_\sigma u^{m+1}_{p_{m+1},\ell}$, then the triple $(u^{m+1}_{p_{m+1},\ell}, u^\ell_{p_\ell,m}, u^\ell_{p_\ell,m+1})\in C^4_\text{ver}$ is not fulfilled by $\sigma$.  If $u^{m+1}_{p_{m+1},\ell} \prec_\sigma u^\ell_{p_\ell,m}$, then the triple $(u^{\ell}_{{p_\ell},m}, u^{m+1}_{p_{m+1},\ell}, u^{m+1}_{p_{m+1},\ell+1})\in C^4_\text{ver}$ is not fulfilled by $\sigma$. 
\end{proof}}
Using \cref{property:sel,property:ver1,property:ver2,property:ver3}, we can prove that the described instance of the General Intermezzo problem is a feasible reduction from the MCP.

\begin{lemma}\label{lemma:w1}
    There is an intermezzo ordering $\sigma$ for $(A,C)$ if and only if $G$ has a multicolored clique of size $k$.
\end{lemma}
\begin{proof}
    First assume that there is an intermezzo ordering $\sigma$. Let the $p_i$ be chosen as in \cref{property:sel}. Then, we define the set $K \subseteq V(G)$ as follows: $K := \{v^i_{p_i}~|~1 \leq i \leq k\}$. \Cref{property:ver2,property:ver3} imply that the set $K$ forms a clique in $G$.

    For the other direction, assume that there is a multicolored clique $K = \{v^1_{p_1}, \dots, v^k_{p_k}\}$ in~$G$. We start our intermezzo ordering in $s^1$. Now we take all the elements of $U^1$ following their ordering implied by $B$ up to $u^1_{p_1,0}$ and then we take $s^2$. We repeat this process for all $i \in \{2,\dots,k\}$. Now we take the $c$-elements following the ordering implied by $B$. Between $c^{\ell, m}$ and $c^{\ell, m+1}$, we take $u^\ell_{p_\ell,m}$ and $u^{m+1}_{p_{m+1},\ell}$ which is possible due to \cref{property:ver3} since $K$ forms a clique. Eventually, we take $c^{k,k}$. Now we first take the remaining elements of $U^1$ in their order in $B$, followed by the remaining elements of $U^2$ and so on. It is easy to check that this ordering is an intermezzo ordering.
\end{proof}
\Cref{thm:w1} follows directly from \cref{lemma:general-to-normal,lemma:w1,lemma:width}, \cref{thm:lower} as well as the \W-hardness of the Multicolored Clique Problem~\cite{fellows2009param,zbMATH02092872}.
Furthermore, \cref{thm:w1} and \cref{lemma:tree-intermezzo} imply the following.

\begin{theorem}\label{thm:w1-tree}
    The \cl-tree recognition problem of Generic Search is \W-hard if it is parameterized by the number $k$ of leaves of the spanning tree. Furthermore, assuming the Exponential Time Hypothesis, there is no algorithm that solves the problem in time $f(k) \cdot n^{o(k)}$ for any computable function $f$ where $n$ is the number of vertices of the given graph.
\end{theorem}


\section{Conclusion}

We have investigated two problems that extend a partial order to a total order while maintaining certain additional constraints. In the first problem, a spanning tree of a graph $G$ is given, which is supposed to be the \cl-tree of Generic Search on~$G$. Surprisingly, deciding this problem turned out to be \NP-complete, although numerous problems involving Generic Search, such as the associated end vertex problem, are straightforward to solve in polynomial time.
This complexity result could be used in the investigation of the second problem. Here we have shown that the General Intermezzo problem cannot be solved in polynomial time even when the Hasse diagram of the given partial order forms a tree. With respect to the width, we were able to specify an \XP-algorithm and at the same time show \W-hardness.

Several questions remain unanswered. For the GS \cl-tree problem it is not clear whether the bounds for the tree height in the \NP-completeness results are best possible. We conjecture that the problem is easy for height 2 and maybe even 3. Furthermore, we suspect that the \NP-completeness also holds for the class of bipartite graphs. While the problem is hard for split graphs, it might be solved efficiently on the subclass of threshold graphs.

Similar questions arise for the (General) Intermezzo problem: We have not shown that the height-bounds for the partial order are best possible. In particular, the bound of~36 in~\cref{corr:36} seems very high. Restricting these partial orders in other ways, e.g. lattices or interval orders, could also be used to find tractable instances of the problem. Furthermore, the complexity status of the Intermezzo problem for cs-trees of bounded height remains open.

\bibliographystyle{plainurl}
\bibliography{lit}

\end{document}